\documentclass[11pt]{article}
\usepackage{fullpage}
\usepackage{url}
\usepackage{xspace}

\usepackage{graphics}
\usepackage[dvips]{epsfig}

\usepackage{amsmath}
\usepackage{amssymb}
\usepackage{amsfonts}
\usepackage{graphicx}
 \usepackage{algorithm}
 \usepackage{algpseudocode}
 \usepackage{multirow}
 \usepackage{subcaption}

\newtheorem{theorem}{Theorem}[section]
\newtheorem{lemma}{Lemma}[section]
\newtheorem{corollary}{Corollary}[section]
\newtheorem{claim}{Claim}[section]

\newtheorem{fact}{Fact}[section]

\newcommand{\qed}{\hfill $\Box$ \bigbreak}
\newenvironment{proof}{\noindent {\bf Proof.}}{\qed}

\newcommand{\cA}{{\cal A}}

\newcommand{\remove}[1]{}

\def\idtt#1{\ensuremath{\mathtt{#1}}}

\def\cross{\idtt{cross}}
\def\up{\idtt{up}}
\def\down{\idtt{down}}

\begin{document}

\baselineskip  0.19in 
\parskip     0.05in 
\parindent   0.3in 

\title{{\bf Sniffing Helps to Meet: Deterministic Rendezvous of Anonymous Agents in the Grid
 }}
\date{}
\newcommand{\inst}[1]{$^{#1}$}

\author{
Younan Gao\inst{1},
Andrzej Pelc\inst{1}$^,$\footnote{Partially supported by NSERC discovery grant 2018-03899 and by the Research Chair in Distributed Computing at the Universit\'e du Qu\'{e}bec en Outaouais.}\\
\inst{1} Universit\'{e} du Qu\'{e}bec en Outaouais, Gatineau, Canada.\\
E-mails: \url{ gaoy03@uqo.ca}, \url{ pelc@uqo.ca}\\
}

\date{ }
\maketitle

\begin{abstract}
Two identical anonymous mobile agents have to meet at a node of the infinite oriented grid whose nodes are unlabeled. 
This problem is known as rendezvous.
The agents execute the same deterministic algorithm.
Time is divided into rounds, and in each round each agent can either stay idle at the current node or move to an adjacent node.
An adversary places the agents at two nodes of the grid at a distance at most $D$, and wakes them up in possibly different rounds. Each agent starts executing the algorithm in its wakeup round. 

If agents cannot leave any marks on visited nodes then they can never meet, even if they start simultaneously at adjacent nodes and know it. Hence, we assume that each agent marks any unmarked node it visits, and that an agent can distinguish if a node it visits has been previously marked or not. (If agents are ants then marking a node means secreting a chemical known as pheromone that can be subsequently sniffed). The time of a rendezvous algorithm is the number of rounds between the wakeup of the later agent and rendezvous. We ask the question whether the capability of marking nodes enables the agents to meet, and if so, what is the fastest rendezvous algorithm.

We consider this rendezvous problem under three scenarios. In the first scenario, agents know $D$ but may start with arbitrary delay. In the second scenario, they 
start simultaneously but do not have any {\em a priori} knowledge. In the third, most difficult scenario, we do not make any of the above facilitating assumptions. Agents start with arbitrary delay and they do not have any a priori knowledge. We prove that in the first two scenarios rendezvous can be accomplished in time $O(D)$. This is clearly optimal.
For the third scenario, we prove that there does not exist any rendezvous algorithm working in time $o(D^{\sqrt{2}})$, and we show an algorithm working in time $O(D^2)$. The above negative result shows a separation between the optimal complexity in the two easier scenarios and the optimal complexity in the most difficult scenario.

\vspace{1ex}

\end{abstract}

\vfill

\vfill

\thispagestyle{empty}
\setcounter{page}{0}
\pagebreak

\section{Introduction}

\noindent
{\bf The background and the problem.}
Two identical anonymous mobile agents have to meet at a node of the infinite oriented grid whose nodes are unlabeled. 
This problem is known as rendezvous.
Agents execute the same deterministic algorithm.
Time is divided into rounds, and in each round, each agent can either stay idle at the current node, or move to an adjacent node.
The adversary places the agents at two nodes of the grid, at distance at most $D$, and wakes them up in 
possibly different rounds. Each agent starts executing the algorithm in its wakeup round. 

If agents cannot leave any marks on visited nodes then they can never meet, even if they start simultaneously at adjacent nodes $u$ and $v$ and know it. Indeed, since agents execute the same deterministic algorithm, the trajectory of one agent is a shift of the trajectory of the other agent by the vector $(u,v)$, and thus agents are never at the same node in the same round.
Hence, we assume that each agent marks any unmarked node it visits, and that an agent can distinguish if a node it visits has been previously marked or not. This is a natural assumption: people walking in the snow leave traces that can be subsequently seen,  ants foraging for food secrete a chemical known as pheromone that can be subsequently sniffed \cite{BBBG,LR}, and territorial animals mark their territory, e.g., by leaving traces of urine, to be noticed by others. 

The time of a rendezvous algorithm is the number of rounds between the wakeup of the later agent and rendezvous. 
We ask the question whether the capability of marking nodes enables the agents to meet, and if so, what is the fastest rendezvous algorithm.

\noindent
{\bf The model.}
Mobile agents navigate in the infinite anonymous oriented grid $\mathbb{Z} \times \mathbb{Z}$. Nodes of the grid do not have labels but ports at each node are coherently labeled $N,E,S,W$. This means that, for any integers $a$ and $b$,  the ports corresponding to edge
$\{(a,b), (a, b+1)\}$ are $N$ at $(a,b)$ and $S$ at  $(a, b+1)$, and ports corresponding to edge $\{(a,b), (a+1, b)\}$ are $E$ at $(a,b)$ and $W$ at $(a+1, b)$.
(The grid is not directed, in the sense that each agent can move in both directions traversing an edge).
Directions corresponding to ports $N,E,S,W$ are called, respectively, North, East, South, West.
Lines East-West of the grid are called {\em horizontal} and lines North-South are called {\em vertical}.
We say that a node $u$ is above (resp. below) a node $v$, if $u$ is North (resp. South) of the horizontal line containing $v$.
We say that a node $u$ is left (resp. right) of a node $v$, if $u$ is West (resp. East) of the vertical line containing $v$.
We use the term {\em distance} to mean the distance in the grid: the distance between nodes $(a,b)$ and $(c,d)$, denoted as
$dist((a,b), (c,d))$, is equal to $|a-c|+|b-d|$.

We do not impose any bound on the memory of the agents. From a computational point of view, they are modelled as Turing machines.
Agents execute the same deterministic algorithm in synchronous rounds.
In each round, each agent can either stay idle at the current node or move to an adjacent node.
The adversary places the agents at two nodes of the grid at a distance at most $D$, and wakes them up in possibly different rounds. Each agent starts executing the algorithm in its wakeup round. The initial node of each agent is called its {\em base}.
If agents cross each other traversing the same edge in opposite directions in some round, they do not even notice this fact.
All the above assumptions are standard in the literature on synchronous rendezvous (cf., e.g.,  \cite{BP,BDL,TSZ07}).

In the beginning, only the bases of the agents are marked. The first agent that visits a node marks it. All marks are identical and non-removable. An agent visiting a node can distinguish whether it is marked or not. Notice that, since each agent remembers its trajectory, even though all marks are identical, an agent visiting an already marked node can infer whether the node was marked by itself or by the other agent. 
A node visited by an agent is {\em foreign} for this agent, if it was already marked at the time of the first visit of this node by this agent. Otherwise, a node is called {\em domestic} for this agent. 
Hence, any agent visiting a marked node can determine if this node is foreign or domestic for it.

The time of a rendezvous algorithm is the number of rounds between the wakeup round of the later agent and the first round in which the agents are at the same node.
If the earlier agent meets the later agent at the starting position of the later agent before its wakeup, we consider rendezvous time to be 0.

\noindent
{\bf Our results.}
We consider the above described rendezvous task under three scenarios. In the first scenario, agents know $D$ but may start with arbitrary delay. In the second scenario, they 
start simultaneously but do not have any {\em a priori} knowledge. In the third, most difficult scenario, we do not make any of the above facilitating assumptions. Agents start with arbitrary delay and they do not have any {\em a priori} knowledge. We prove that in the first two scenarios rendezvous can be done in time $O(D)$. This is clearly optimal.
For the third scenario we prove that there does not exist a rendezvous algorithm working in time $o(D^{\sqrt{2}})$, and we show an algorithm working in time $O(D^2)$.
The above negative result shows a separation between the two easier scenarios and the most difficult scenario, regarding the optimal complexity of rendezvous.

In all scenarios, the main difficulty that has to be overcome by a rendezvous algorithm is symmetry breaking. Since agents are identical and execute the same deterministic algorithm, the parts of their trajectories before visiting a foreign node  are shifts of each other  (this is the reason why deterministic rendezvous of anonymous agents is impossible without marking nodes, even with simultaneous start). The only way to guarantee rendezvous is to design rules of behavior of an agent based on the information in which rounds  it visited its foreign nodes and in which direction it moved in the rounds of such visits. We design these rules so that, regardless of the relative positions of the bases of the agents (which are decided by the adversary), one agent stops at some node of the trajectory of the other agent which then meets it at this node. Moreover, the meeting should occur fast. The challenge is to design the rules in such a way that this property be satisfied for any relative positions of the bases, without knowing them. This design and the proof of correctness and complexity of the resulting rendezvous algorithms are the main algorithmic contributions of this paper. Another contribution, implied by our negative result, is showing that the optimal complexity of rendezvous in the most difficult scenario is strictly larger than in any of the easier ones.

 \noindent
{\bf Related work.}
The task of rendezvous in graphs has been extensively investigated in distributed computing literature, both under randomized and deterministic scenarios.
It is usually assumed that nodes do not have distinct identities, and agents cannot mark nodes, although rendezvous was also considered in labeled graphs \cite{CCGKM,MP}, or when marking nodes by agents using tokens is allowed \cite{KKM,KKSS}.
A survey of  randomized rendezvous in various models  can be found in the classic book
\cite{alpern02b}. 
Deterministic rendezvous in graphs has been surveyed in \cite{Pe2}.
 
Most of the literature on rendezvous considered finite graphs and assumed the synchronous scenario, where
agents move in rounds. 
 In \cite{TSZ07}, the authors presented rendezvous algorithms with time polynomial in the size of the graph and the length of agents' labels.
Gathering many agents in the presence of Byzantine agents was discussed in \cite{BDL}.
The amount of memory required by the agents to achieve deterministic rendezvous was investigated in  \cite{CKP} for arbitrary finite graphs.

%

Fewer papers were devoted to synchronous rendezvous in infinite graphs. In all cases, agents had distinct identities but could not mark the visited nodes.
In \cite{CCGKM}, the authors considered rendezvous in infinite trees and grids,  using the strong assumption that the agents know their location in the environment. In \cite{BP2}, rendezvous in infinite trees was considered. In \cite{BP}, the authors presented a rendezvous algorithm working for arbitrary, even infinite, graphs with nodes of finite degrees. Rendezvous for arbitrary graphs, including those with nodes of infinite degree, was studied in \cite{PP}.

Among infinite graphs, the infinite oriented grid (considered in this paper) plays a special role in the context of the rendezvous problem because it is closely related to the problem of {\em approach in the plane}, in which agents equipped with compasses start in arbitrary points of the plane, and have to get at a distance at most 1. Rendezvous in the grid implies approach in the plane, hence it was studied by many authors under a variety of scenarios \cite{BCGIL,BP,BBDDP,CCGKM,DP}. 

Several authors investigated asynchronous rendezvous in the plane \cite{CFPS,fpsw} and 
in graphs
	\cite{BBDDP,BCGIL,DPV}.
	In the latter scenario, the agent chooses the edge to traverse, but the adversary controls the speed of the agent. Under this assumption, rendezvous
	at a node cannot be guaranteed even in the two-node graph. Hence, in the asynchronous scenario, the agents are permitted to meet inside an edge.
	In \cite{BCGIL}, the authors designed almost optimal algorithms for asynchronous rendezvous in infinite multidimensional grids, under a strong assumption that an agent knows its position in the grid. In \cite{BBDDP,DPV} this assumption was replaced by a much weaker assumption that agents have distinct identities. In \cite{BBDDP}, a polynomial-cost algorithm was designed for the infinite oriented two-dimensional grid, and in \cite{DPV} -- for arbitrary finite graphs.

\section{Preliminaries and Terminology}

When an agent visits a foreign node for the first time, we say that it made a $Dir$-{\em hit}, where $Dir\in \{N,E,S,W\}$, if in the round of this visit it went in direction $Dir$. Thus, for example,  if at the first visit of a foreign node the agent was going North, we say that it made an $N$-hit. A hit is called {\em horizontal} if it is an $E$-hit or a $W$-hit, and it is called {\em vertical} if it is an $N$-hit or an $S$-hit.
A set of hits is called {\em homogeneous}, if they are all $Dir$-hits, for the same direction $Dir$. Otherwise, it is called {\em heterogeneous}.

In all our algorithms, we will use instruction {\tt go} $Dir$ {\tt for} $x$ {\tt rounds}, where $Dir\in \{N,E,S,W\}$ and $x$ is a positive integer. It tells
the agent to make $x$ steps in direction $Dir$ in $x$ consecutive rounds. 

In two of our algorithms we use the following procedure $\cross(x)$, defined for any positive integer parameter $x$.

\noindent
{\bf Procedure} $\cross(x)$

\noindent
{\tt go} $W$ for $x$ rounds;
{\tt go} $E$ for $x$ rounds;\\
{\tt go} $N$ for $x$ rounds;
{\tt go} $S$ for $x$ rounds;\\
{\tt go} $E$ for $x$ rounds;
{\tt go} $W$ for $x$ rounds;\\
{\tt go} $S$ for $x$ rounds;
{\tt go} $N$ for $x$ rounds;\\

\section{Known Upper Bound on Initial Distance}\label{sec:known}

In this section, we consider the scenario of the rendezvous problem, in which both agents know a common upper bound $D$ on their initial distance  but might start in different rounds. Our aim is to present and analyze an algorithm that guarantees rendezvous in time $O(D)$ under this scenario.

\subsection{The algorithm}
In this subsection, we present the algorithm, while the proof of its correctness and complexity is deferred to the next subsection.

Algorithm {\tt Known Upper Bound} consists of two parts. The first part is the execution of 
procedure $\cross(D)$.
During this execution, the agent stores the following information: 
the nodes where it makes hits and the directions of the moves when these hits occur.

Upon completion of procedure $\cross(D)$ the agent is back at its base and initiates the second part.
The second part of the algorithm is one of the two following actions: 
\begin{itemize}
	\item Action $I$: Repeat procedure $\cross(D)$ until rendezvous.
	\item Action $II$: Move to the earliest visited foreign node and stay there forever.	
	\end{itemize}
The decision which of these actions has to be chosen is given in the following table. The input in the first column of the table stores the information recorded by the agent after completing procedure $\cross(D)$ in the first part of the algorithm. The action corresponding to a given input is in the second column.  Note that all inputs listed in the table are pairwise exclusive, hence the choice of the action is unambiguous.

	\begin{table}[!h]
\begin{center}
	\begin{tabular}{ |c|c| } 
		\hline
		Input & Action		\\
		\hline
		\hline
		No hits & $I$  \\ 
		\hline
		Exactly two hits, with
		a horizontal one followed by a vertical one, or vice versa & $II$  \\ 
		\hline
		Only $S$-hits (an arbitrary positive number of them) & $II$  \\ 
		\hline
		Only $N$-hits (an arbitrary positive number of them) & $I$  \\ 
		\hline
		Only horizontal hits, starting with a $E$-hit & $II$\\
		\hline
		Only horizontal hits, starting with a $W$-hit & $I$\\
		\hline
		Only vertical but heterogeneous hits and condition $C$ & $I$ \\
		\hline
		Only vertical but heterogeneous hits and condition $\neg C$ & $II$\\
		\hline
	\end{tabular}
		\caption{\label{tab-decision-1}The decision table for Algorithm {\tt Known Upper Bound}.  C is the following condition: An $N$-hit happens at the North-most visited node.   }
\end{center}
\end{table}

\subsection{Correctness and complexity}\label{cor-cor-bounded}

The high-level idea of the proof of correctness is to show that,  for any possible relative positions of the bases of the agents, each of the agents must get, upon completing procedure $\cross(D)$, one of
the 8 types of inputs enumerated in the decision table of the algorithm, and choose a different action.
This is enough to guarantee rendezvous. Indeed, 
let $a_I$ and $a_{II}$ be the agents that perform actions $I$ and $II$, respectively.
After performing action $II$, agent $a_{II}$ stays at a node $u$ of the trajectory of $a_I$, performed in the execution of procedure $\cross(D)$ by $a_I$.
On the other hand, $a_{I}$ keeps performing $\cross(D)$ until rendezvous. Upon at most one execution of procedure $\cross(D)$ performed by $a_{I}$ after $a_{II}$ stabilizes at $u$, rendezvous must occur.

Observe that there are three possible types of relative positions of the bases of the agents. 
\begin{itemize}
\item both bases are not on the same line of the grid, 
\item both bases are on the same horizontal line 
\item both bases are on the same vertical line.
\end{itemize}

The proof of correctness of Algorithm {\tt Known Upper Bound} is split into three parts corresponding to these types of positions.
In the sequel, we call one of the agents $B$ (blue) and the other agent $R$ (red).

\subsubsection{Both bases are not on the same line}

\begin{theorem}\label{th:different lines}
If the bases of the agents are not on the same line then the agents choose different actions in the second part of Algorithm {\tt Known Upper Bound}.
\end{theorem}

\begin{proof}
If the bases of both agents are not on the same line,
only two hits can happen because the trajectories of the agents have only two common nodes.
Either both hits are made by the same agent or each agent makes exactly one hit.
We will consider both cases separately.
We start with the following claim.

\begin{claim}
	\label{lem-hori-vert}
Both hits are made by the same agent if and only if one hit is horizontal and the other one is vertical.
\end{claim}

In order to prove the claim,
let $x$ and $y$ be the common nodes of the trajectories of both agents. The ``only if'' direction of the equivalence follows from the fact that the base of the agent that makes the two hits is on the same vertical line as one of the nodes $x$ and $y$, and on the same horizontal line as the other (cf. Fig. \ref{fig:diagonal}).
	
	\begin{figure}[!h]
		\centering
		\includegraphics[scale=1]{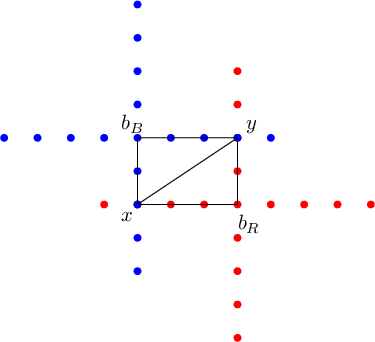}
		\caption{Trajectories of both agents after completing procedure $\cross(4)$. $x$ and $y$ are the nodes where the two hits are made. The bases, along with $x$ and $y$, form a rectangle. Both $x$ and $y$ are on the same diagonal.}
		\label{fig:diagonal}
	\end{figure}

Next, we prove the ``if '' part of the equivalence:  i.e, the implication that if one hit is horizontal and the other one is vertical, then both hits have to be made by the same agent.
Observe that the bases of both agents along with $x$ and $y$ form a rectangle, in which $x$ and $y$ are on one diagonal of the rectangle and both bases are on the other diagonal.
Without loss of generality, we assign $x$ to the corner of the rectangle that shares the same vertical line with the base of agent $B$. Therefore, $x$ and the base of $R$ are on the same horizontal line.
Hence $y$ shares the same vertical line with the base of agent $R$ and the same horizontal line with the base of $B$.
If a vertical hit happens at $x$ and a horizontal hit happens at $y$, then both $x$ and $y$ are domestic nodes of $R$. Therefore, both hits are made by $B$.
If a horizontal hit happens at $x$ and a vertical hit happens at $y$, then both $x$ and $y$ are domestic nodes of $B$.
Therefore, both hits are made by $R$.
In either case, both hits are made by the same agent.
This proves the claim. $\diamond$

\noindent
{\bf Case 1.} Both hits are made by the same agent.

In view of Claim \ref{lem-hori-vert}, the agent that made both hits has a vertical hit and a horizontal hit in its input. 
Therefore, in view of the decision table, it chooses action $II$.
The other agent did not make any hits, and hence it chooses action $I$. Both agents choose different actions.

\noindent
{\bf Case 2.} Each agent makes exactly one hit.

\begin{claim}
	\label{lem-hit-exact-one-each}
If each agent makes exactly one hit, then either one hit is an $E$-hit and the other one is a $W$-hit or one hit is an $N$-hit and the other is an $S$-hit.
\end{claim}

To prove the claim, observe that
	in view of Claim \ref{lem-hori-vert}, if each agent makes exactly one hit, then either both hits are horizontal or both hits are vertical.
	Both hits cannot be homogeneous because  they are made at nodes $x$ and $y$ which are on the diagonal of the rectangle formed by them and by the bases of the agents (cf. Fig. \ref{fig:diagonal}). This proves the claim. $\diamond$

	Consider the two possibilities following from Claim \ref{lem-hit-exact-one-each}. In view of the decision table,  the agent that makes an $E$-hit or an $S$-hit chooses action II in the second part of the algorithm, and the agent that makes a $W$-hit or an $N$-hit chooses action I.
Both agents choose different actions.
\end{proof}

\subsubsection{Both bases are on the same horizontal line}

\begin{theorem}\label{th:horizontal}
	If the bases of both agents are on the same horizontal line, then both agents choose different actions in the second part of Algorithm {\tt Known Upper Bound}.
\end{theorem}

\begin{proof}
	We may assume that rendezvous does not happen by the round in which the adversary wakes up the later agent. Otherwise, according to the definition of rendezvous time, the meeting occurs in time 0.
	Let $B$ denote the agent whose base is West of the base of the other agent and let $R$ denote the other agent. 
	Notice that each agent must make at least one hit.
	We start with the following claim.
	
	\begin{claim}
		\label{claim-horizontal}
		Agent $B$ always makes an $E$-hit as its first hit, while agent $R$ always makes a $W$-hit as its first hit.
	\end{claim}
	
	Recall that when each agent is woken up, it leaves immediately a mark at its base.
	Neither agent $R$ has arrived at the base of agent $B$,  by the round when agent $B$ is woken up, nor agent $B$ has arrived at the base of agent $R$, by the round when agent $R$ is woken up; otherwise, rendezvous happens by the round when the adversary wakes up the later agent.
	Let $t$ denote the round in which agent $B$ is woken up by the adversary.
	Since agent $R$ has not yet arrived at the base of agent $B$ by round $t$, all nodes that are West of the base of $B$ are unmarked by round $t$.
	Agent $B$ cannot make any hits as long as it is going West.
	Hence, no $W$-hits are made by $B$.
	On the other hand, when $B$ is going East, it makes at least one $E$-hit at the base of $R$.
	Therefore, the first hit made by $B$ is always an $E$-hit.
	Once agent $R$ is woken up by the adversary, it goes West before going in the other three directions. 
	When $R$ is going West, it makes at least one $W$-hit at the base of agent $B$.
	Therefore, agent $R$ always makes a $W$-hit as its first hit.
	This completes the proof of the claim. $\diamond$
	
	In view of the decision table, agent $B$ chooses action $II$, while agent $R$ chooses action $I$, in the second part of algorithm {\tt Known Upper Bound}.
	Hence each agent chooses a different action.
\end{proof}

\subsubsection{Both bases are on the same vertical line}

We now assume that the bases of both agents are on the same vertical line $L$.
As before, we want to show that both agents choose different actions.
As we will see, the argument in the present case is more involved than in the case when both bases are on the same horizontal line.
The reason is that the horizontal and vertical directions do not play the same role in our algorithm, due to the fact that procedure $\cross(D)$ starts by going West.

\begin{theorem}\label{th:vertical}
	If the bases of both agents are on the same vertical line, denoted as $L$, then both agents choose different actions in the second part of Algorithm {\tt Known Upper Bound}.
\end{theorem}

\begin{proof}
Recall that each agent marks its own base immediately after the adversary wakes it up.
Hence, the base of an agent is always a domestic node of this agent.
In this proof we will only consider nodes on line $L$.

Given a node $x$ of the grid, we denote by $\up(x)$ the Northern neighbor of $x$ and by $\down(x)$ the Southern neighbor of $x$.
Let $B$ denote the agent whose base is North of the base of the other agent and let $R$ denote the other agent.
We denote by $top$ (resp. $bottom$) the North-most (resp. South-most) node visited by any agent, by the round in which the later agent completes procedure $\cross(D)$.
Note that both nodes $bottom$ and $top$ are unique.
We use $b_{x}$ to represent the base of $x$, where $x\in \{B, R\}$.
There are four possibilities to consider.
\begin{itemize}
\item Node $\up(b_B)$ is domestic for $B$ and node $\down(b_R)$ is domestic for $R$.
\item Node $\up(b_B)$ is domestic for $B$ and node $\down(b_R)$ is domestic for $B$.
\item Node $\up(b_B)$ is domestic for $R$ and node $\down(b_R)$ is domestic for $R$.
\item Node $\up(b_B)$ is domestic for $R$ and node $\down(b_R)$ is domestic for $B$.
\end{itemize}
Our goal is to show that for each of these four potential possibilities, either it can never occur or both agents choose different actions.
We start with the following claim.

\begin{claim}
	\label{fact-marking}
	i) Node $top$ is domestic for $B$; ii) node $bottom$ is domestic for $R$; iii) on the segment between $b_{B}$ and $b_{R}$, all domestic nodes of $B$ are North of all domestic nodes of $R$.
	\end{claim}
	
We only prove statement (iii), since (i) and (ii) are obvious.
Assume that there exists a pair of nodes $x$ and $y$ on the segment between $b_B$ and $b_R$ such that $x$ is domestic for $B$, $y$ is domestic for $R$ and $x$ is South of $y$.
Since $x$ is domestic for $B$, we know that $B$ visits $x$ for the first time before $R$ does.
Let $t$ denote the round in which $B$ visits $x$ for the first time.
Since $y$ is between $b_B$ and $x$, agent $B$ arrives at $y$ before round $t$.
Note that $x$ is North of $b_R$. 
By round $t$, all domestic nodes of $R$ are South of $x$, since $R$ has not arrived at $x$ by round $t$.
Therefore, agent $B$ visits $y$ before $R$ does, and node $y$ cannot be a domestic node of $R$.
This is a contradiction which proves statement (iii). $\diamond$

	In Claim \ref{fact-marking}, we discussed nodes in the segment between $b_B$ and $b_R$, which is part of the segment between $top$ and $bottom$.
	Indeed, the latter segment can be divided into three sub-segments, i.e., the segment between $top$ and $b_{B}$, the segment between $b_{B}$ and $b_{R}$, and the segment between $b_{R}$ and $bottom$.
	Next, we show the properties of the other two sub-segments.
	
	\begin{claim}
		\label{fact-homo}
If  $\up(b_B)$ is a domestic node of $B$, then all nodes between $top$ and $b_B$ are domestic nodes of $B$.
If $\down(b_R)$ is a domestic node of $R$, then all nodes between $b_R$ and $bottom$ are domestic nodes of $R$.
	\end{claim}

We only prove the first statement, as the proof of the second is similar.	
Since node $\up(b_B)$ is a domestic node of $B$, agent $R$ has not arrived at $\up(b_B)$ when $B$ arrives there for the first time.
By the round in which $B$ arrives at $\up(b_B)$ for the first time, none of the nodes North of $\up(b_B)$ has been marked.
So $B$ marks these nodes when it visits each of them for the first time.
Therefore, all nodes between $b_{B}$ and $top$ are domestic nodes for $B$.
This proves the claim. $\diamond$

	Now, we are ready to prove that agents $B$ and $R$ always choose different actions in any of the four above mentioned possibilities that may occur.	
	
	
	\noindent
	{\bf Case 1. $\up(b_B)$ is a domestic node of $B$ and  $\down(b_R)$ is a domestic node of $R$.}
	
In view of Claim \ref{fact-homo}, all nodes between $b_B$ and $top$ are domestic nodes of $B$, while all nodes between $b_R$ and $bottom$ are domestic nodes of $R$.
	Furthermore, on the segment between $b_B$ and $b_R$, all domestic node of $B$ are North of all domestic nodes of $R$, in view of part (iii) of  Claim \ref{fact-marking}.
	Therefore, on the segment between $top$ and $bottom$, all domestic nodes of $B$ are North of all domestic nodes of $R$.
	Hence, agent $B$ makes only $S$-hits, while agent $R$ makes only $N$-hits.
	Note each agent makes at least one hit, at the base of the other agent.
	Therefore, in view of the decision table, agent $B$ chooses action II, and agent $R$ chooses action I. 
	
	\noindent
	{\bf Case 2. $\up(b_B)$ is a domestic node of $B$ and $\down(b_R)$ is a domestic node of $B$.}

Since $\up(b_B)$ is a domestic node of $B$, so are the nodes between $b_B$ and $top$, in view of Claim \ref{fact-homo}.
So $B$ cannot make any $N$-hits.
Agent $B$ makes at least one $S$-hit at $b_R$, and thus in view of the decision table, it chooses
action II.

Recall that agent $R$ makes at least one $N$-hit at $b_B$.
	Let $t$ denote the earliest round in which $R$ makes an $N$-hit, and let $u$ denote the node where this first $N$-hit was made.
	It follows that $u$ is not North of $b_B$.
	
	\begin{claim}\label{North of u}
	All nodes that are North of $u$ are domestic for $B$.
	\end{claim}
	
	In order to prove the claim, notice that
	since $u$ is domestic for $B$, agent $B$ visits $u$ for the first time before $R$ does.
	If $u$ is South of $b_B$, then in the round when $B$ visits $u$ for the first time, $B$ was going South.
	In this case, by round $t$, in which $R$ was going North and arrived at $u$ for the first time, all nodes that are North of $u$ have already been visited by $B$, and thus they are domestic for $B$.
	Therefore, if $u$ is South of $b_B$, the claim is true.
	Otherwise, $u=b_B$.
	Note that $\up(b_B)$ is a domestic node of $B$.
	Hence, all nodes between $u$ and $top$ are domestic for $B$, in view of Claim \ref{fact-homo}.
	This proves the claim. $\diamond$

	In view of Claim \ref{North of u}, when $R$ visits any of the nodes North of $u$ for the first time, including the North-most node $w$ visited by $R$,  an $N$-hit is made.
	Therefore, condition $C$ in the decision table applies.
	After completing the execution of $\cross(D)$, $R$ has seen $N$-hits first and then at least one $S$-hit at $\down(b_R)$.
	In particular, among those $N$-hits, the last one was made at node $w$.
	Therefore, in view of the decision table, $R$ chooses action $I$, upon completion of procedure $\cross(D)$.
	Hence the agents choose different actions.

	\noindent
	{\bf Case 3. $\up(b_B)$ is a domestic node of $R$ and $\down(b_R)$ is a domestic node of $R$.}
	
	Since node $\down(b_R)$ is domestic for $R$, so are the nodes between $b_R$ and $bottom$, in view of Claim \ref{fact-homo}.
	So, agent $R$ cannot make any $S$-hits.
	As we know, $R$ makes at least one $N$-hit at $b_B$. Hence, agent $R$ chooses action I, in view of the decision table.
	
	Since $\up(b_B)$ is domestic for $R$, an $N$-hit is made by $B$ when $B$ arrives at $\up(b_B)$ for the first time.
	In view of part i) of Claim \ref{fact-marking}, node $top$, as the North-most node ever visited by $B$, is domestic for $B$.
	Agent $B$ cannot make a hit at $top$, so condition $\neg C$ applies.
	
	\begin{claim}
	Node $\down(b_B)$ is domestic for $R$.
	\end{claim}
	
	In order to prove the claim, notice that,
	before agent $B$ starts moving vertically, agent $R$ has arrived at $\up(b_B)$.
	Note that agent $R$ arrives at $\down(b_B)$ for the first time before arriving at $\up(b_B)$ for the first time.
	We know that agent $R$ arrives at $\down(b_B)$ for the first time before agent $B$ does.
	Therefore, $\down(b_B)$ is domestic for $R$. This proves the claim.$\diamond$

	Upon completion of $\cross(D)$, agent $B$ has made at least one $N$-hit at $\up(b_B)$, was the only one to visit $top$, and then made at least one $S$-hit at $\down(b_B)$ because $\down(b_B)$ is domestic for $R$.
	In view of the decision table, agent $B$ chooses action II.
	Therefore, both agents choose different actions.

	\noindent
	{\bf Case 4. $\up(b_B)$ is a domestic node of $R$ and $\down(b_R)$ is a domestic node of $B$.} 
	
	It turns out that this case can never happen. As we have seen in the previous case, when agent $R$ arrives at $\up(b_B)$ for the first time, agent $B$ has not yet started moving vertically.
	Hence agent $R$ was woken up earlier than agent $B$, so agent $R$ starts moving South before agent $B$ does.
	Note that $\down(b_R)$ is the first node that agent $R$ visits after it starts moving South.
	Therefore, agent $R$ arrives at $\down(b_R)$ for the first time before agent $B$ does.
	It follows that node $\down(b_R)$ cannot be domestic for $B$, which is a contradiction.
	
	As we have seen, in all possible cases when the bases of both agents are on the same vertical line, the agents choose different actions in the second part of the algorithm. 
	This completes the proof.
\end{proof}

Theorems \ref{th:different lines}, \ref{th:horizontal} and \ref{th:vertical} imply that, for all possible relative positions of the bases of the agents, the agents choose different actions in the second part of Algorithm {\tt Known Upper Bound}. As noted in the beginning of this section, this implies rendezvous.

It remains to analyze the complexity of Algorithm {\tt Known Upper Bound}.

\begin{lemma}\label{compl Known Upper Bound}
The execution time of Algorithm {\tt Known Upper Bound} is in $O(D)$, if the agents start at two nodes at a distance at most $D$.
\end{lemma}

\begin{proof}
According to the definition, counting time starts in the wake-up round of the later agent. Call this round 0.
Let $a_I$ denote the agent that chooses action $I$, and let $a_{II}$ denote the agent that chooses action $II$.
Procedure $\cross(D)$ uses $8D$ rounds.
Therefore, by the end of round $8D$, both agents have decided which action they should take in the second part of the algorithm.
Within the next $2D$ rounds, agent  $a_{II}$
arrives at one of the nodes on the trajectory of agent $a_{I}$, and stops. Within another $8D$ rounds, agent $a_{I}$
that repeats procedure $\cross(D)$ finds agent $a_{II}$ on its trajectory.
Hence, the algorithm uses no more $8D+2D+8D = 18D$ rounds, which is in $O(D)$.
\end{proof}

Theorems  \ref{th:different lines}, \ref{th:horizontal}, \ref{th:vertical} and Lemma \ref{compl Known Upper Bound} imply the following corollary.

\begin{corollary}
Suppose that both agents know the same upper bound $D$ on their initial distance. Then Algorithm {\tt Known Upper Bound} guarantees rendezvous in time $O(D)$.
\end{corollary}

\section{Simultaneous Start}

In this section, we consider the second scenario of the rendezvous problem, in which both agents are simultaneously woken up by the adversary but they do not have any {\em a priori} knowledge.
The adversary places both agents at two nodes of the grid at a distance at most $D$.
We will design  a deterministic algorithm that guarantees rendezvous in time $O(D)$.
The algorithm is presented in Section \ref{sect-alg-sim}, and the proof of its correctness and complexity is deferred to Section \ref{sect-sim-proof}.

\subsection{The Algorithm}
\label{sect-alg-sim}

Algorithm {\tt Simultaneous Start} works in phases numbered $0,1, \dots$ performed until rendezvous is achieved.
In the beginning of the $i$-th phase, where $i\ge 0$, each agent chooses one of the two following actions: 
\begin{itemize}
	\item Action $I$: Perform procedure $\cross(2^i)$.
	\item Action $II$: Move to the earliest visited foreign node and stay there forever (in this and all future phases).
%
%
\end{itemize}
In each phase, if an agent chooses action I, it collects the same information as in Algorithm {\tt Known Upper Bound}:
the nodes where it makes hits and the directions of the moves when these hits occur.

The decision which of these actions has to be chosen in a given phase $i$ is presented in Table \ref{tab-decision-sim}. The first column records all types of information that an agent could collect in phase $i-1$, if it executed action $I$ in this phase. This information is used as input to determine the action chosen in phase $i$. The second column shows the actions to be taken in phase $i$ that correspond to each possible input.
Note that all inputs listed in the table are pairwise exclusive, hence the choice of the action is unambiguous.

\begin{table}[!h]
	\begin{center}
		\begin{tabular}{ |c|c| } 
			\hline
			Input & Action		\\
			\hline
			\hline
			No hits & $I$  \\ 
			\hline
			Exactly two hits, with
			a horizontal one and a vertical one & $II$  \\ 
			\hline
			Only $S$-hits (an arbitrary positive number of them) & $II$  \\ 
			\hline
			Only $N$-hits (an arbitrary positive number of them) & $I$  \\ 
			\hline
			Only $E$-hits (an arbitrary positive number of them) & $II$\\
			\hline
			Only $W$-hits (an arbitrary positive number of them)  & $I$\\
			\hline
		\end{tabular}
		\caption{\label{tab-decision-sim}The decision table for Algorithm {\tt Simultaneous Start} }
	\end{center}
\end{table}

\subsection{Correctness and complexity}
\label{sect-sim-proof}
The high-level idea of the proof of correctness of Algorithm {\tt Simultaneous Start}  is to show that each agent will get, in the beginning of each phase, one of the six types of inputs enumerated in the decision table of the algorithm, and that by phase $\lceil \log D \rceil+1$, the input of each agent corresponds to a different action.
This is enough to guarantee that agents will meet.
Indeed, let $p$ denote the earliest phase in which both agents choose different actions, and let $a_I$ and $a_{II}$ denote the agents that perform actions $I$ and $II$, respectively, in phase $p$.
In phase $p$, agent $a_{II}$ moves, after at most $2^p$ rounds, to a node $u$ in the trajectory of $a_I$, performed in the execution of procedure $\cross(2^{p})$, and stays there forever. Agent 
$a_I$ executes procedure $\cross(2^{p})$ in phase $p$, and makes its second visit at node $u$ after more than $2^p$ rounds.
Hence, both agents meet at node $u$
by the end of phase~$p$.

Observe that there are two possible types of relative positions of the bases of the agents: either both bases are on the same line or not. 
We prove that $p\leq \lceil \log D \rceil+1$, in each of these cases.

\begin{theorem}\label{dist-line}
If the bases of the agents are at a distance $d$ and they are on the same line then agents meet at the latest in phase $ \lceil \log d\rceil+1$.
\end{theorem}
\begin{proof}
We give the proof in the case in which the bases of both agents are on the same horizontal line. The proof when they are on the same vertical line is similar. Let $B$ denote the agent whose base $b_B$ is West of the base $b_R$ of the other agent, and let $R$ denote the other agent.
We start with the following claim.

\begin{claim}
	\label{lem-hori-exlusive-W-E}\label{hor}
	Agent $B$ cannot make any $W$-hits and agent $R$ cannot make any $E$-hits.
\end{claim}
In order to prove the claim, note that
	each node that is West of $b_B$ is visited by $B$ before being visited by $R$, since $b_B$ is West of $b_R$.
	Therefore, no nodes that are West of $b_B$ are domestic for $R$.
	Similarly, each node that is East of $b_R$ is visited by $R$ before being visited by $B$.	
	Hence, no nodes that are East of $b_R$ are domestic for $B$. 
	As a result, neither $B$ makes any $W$-hits nor $R$ makes any $E$-hits.$\diamond$

	Observe that, both agents make hits in the same phase, and that
	by phase $\lceil \log d \rceil$, agent $B$ has visited the base $b_R$ of $R$, at which $B$ made an $E$-hit, and agent $R$ has visited the base $b_B$ of $B$, at which $R$ made a $W$-hit. Let $p'$  be the phase in which the first hit is ever made. Thus $p'\leq \lceil \log d \rceil$.
Let $p=p'+1$.
In phase $p$, agent $B$ chooses action $II$, while agent $R$ chooses action $I$, in view of  Claim \ref{hor} and of the decision table.
Hence agent $B$ goes to some node of the trajectory of $R$ at which $R$ meets it. 
\end{proof}

\begin{theorem}\label{dist-no-line}
If the bases of the agents are at a distance $d$ and they are not on the same line then agents meet at the latest in phase $ \lceil \log d\rceil+1$.\end{theorem}
\begin{proof}
We start with the following claim.

\begin{claim}
	\label{lem-two-critical-points}
	In any phase, either no hits are made or both hits are made. 
\end{claim}

We prove the claim by contradiction.
	Let $x$ and $y$ denote the two nodes where the two hits are made, such that $x$ and the base $b_B$ of $B$ share the same vertical line and $y$ and the base $b_R$ of $R$ share the same vertical line (cf. Fig. \ref{fig:diagonal}).
	Observe that nodes $x$, $y$, $b_B$ and $b_R$ form a rectangle.
	It follows that $dist(x, b_B)=dist(y, b_R)$ and $dist(x, b_R)=dist(y, b_B)$.
	
	Let $p'$ denote the phase in which the first hit was ever made.
	So, in all previous phases, no hits were made.
	Assume that in phase $p'$, only the hit at node $x$ was made.
	Following this assumption, we know that both agents visit $x$ in phase $p'$, while at least one agent has not visited node $y$.
	First, we consider the case in which agent $B$ has not visited $y$ in phase $p'$.
	Let $u$ denote the node that is closest to $y$, among all the nodes that are on the trajectory of $B$.
	Then, it follows that $2^{p'}=dist(b_B, u)<dist(b_B, y)=dist(b_R, x)$.
	On the other hand, since agent $R$ has visited node $x$ in phase $p'$, we have $dist(b_R, x)\le 2^{p'}$, which leads to a contradiction.
	Next, we consider the case in which agent $R$ has not visited $y$ in phase $p'$.
	Let $v$ denote the node that is closest to $y$, among all the nodes that are on the trajectory of $R$.
	Then, it follows that $2^{p'}=dist(b_R, v) < dist(b_R, y)=dist(b_B, x)$.
	On the other hand, since agent $B$ has visited $x$, we have $dist(b_B, x)\le 2^{p'}$.
	This leads to a contradiction.
		Hence, it is impossible to have only one hit made at node $x$ in phase $p'$.
			The assumption that only one hit at node $y$ was made in phase $p'$ can be refuted in a similar way.
		Therefore, both hits must be made in phase $p'$.
		Note that since the trajectories of the agents have only two common nodes, once two hits were made in phase $p'$, no hits can be made in the future phases. $\diamond$
	
		Let $p'$ denote the phase in which the first hit was ever made. Then $p'\leq  \lceil \log d\rceil$.
	In view of Claim \ref{lem-two-critical-points}, two hits are made in phase $p'$. Either both hits are made by the same agent or each agent makes exactly one hit.
	Let $p=p'+1$.
	If both hits are made by the same agent, then one hit is horizontal and the other one is vertical.
	In view of Table \ref{tab-decision-sim}, the agent that made two hits in phase $p'$ chooses action $II$ in phase $p$.
	The other agent made no hits in phase $p'$, so it chooses action $I$ in phase $p$.
	Therefore, the agents choose different actions in phase $p$.
	If each agent made only one hit, then either one hit is an $E$-hit and the other one is a $W$-hit, or one hit is an $N$-hit and the other is an $S$-hit.
	In either case, both agents choose different actions in phase $p$. Since 
	 $p\leq  \lceil \log d\rceil +1$, the proof is concluded.
\end{proof}

Theorems \ref{dist-line} and \ref{dist-no-line} imply the correctness of Algorithm {\tt Simultaneous Start}.
It remains to analyze its complexity.

\begin{theorem}\label{compl Simultaneous Start}
Suppose that the agents start simultaneously at a distance at most $D$. Then Algorithm {\tt Simultaneous Start} guarantees rendezvous in time $O(D)$.
\end{theorem}

\begin{proof}
Let $p$ be the last phase of Algorithm {\tt Simultaneous Start}. 
Since the agents start at a distance at most $D$,
Theorems \ref{dist-line} and \ref{dist-no-line} imply that  $p\leq  \lceil \log D\rceil+1$. 
For any $i<p$, the duration of phase $i$ is $8 \cdot 2^i$, and the duration of phase $p$ is at most $8 \cdot 2^p$.
Hence the duration of the entire algorithm is at most $8\cdot 2^{p+1}$. Since $p\leq  \lceil \log D\rceil+1$, the complexity of the algorithm is in $O(D)$.
\end{proof}

\section{Arbitrary Delay with No Knowledge}

In this section we do not make any of the facilitating assumptions made in the previous sections: agents start with arbitrary delay and they do not have any a priori knowledge. In particular, unlike in Section \ref{sec:known}, they do not know any upper bound on their initial distance. We prove that, in this most difficult scenario,  rendezvous time complexity is strictly larger than in any of the two previously considered cases.  Now rendezvous cannot be achieved in time $O(D)$. This is implied by the following result.

\begin{theorem}
Let $\cA$ be any deterministic rendezvous algorithm. There exists an infinite sequence of positive integers $(d_0,d_1,d_2,\dots)$ diverging to infinity,  and a positive constant $c$, such that, for any $i\geq 0$, there exist two nodes of the grid at a distance $d_i$, and an integer $\delta_i$,  that satisfy the following property: agents woken up at these nodes with delay $\delta_i$ and executing algorithm $\cA$ must use time at least $cd_i^{\sqrt{2}}$ to accomplish rendezvous. 
\end{theorem}

\begin{proof}
Fix a deterministic rendezvous algorithm $\cA$. Consider the infinite trajectory $T$ of the agent that is alone in the grid, starts at node $v$ and executes algorithm $\cA$. Hence $T$ is an infinite sequence of nodes of the grid.


First observe that the trajectory $T$ must contain nodes that are arbitrarily far from $v$. Otherwise, for some positive integer $y$, all nodes  of the trajectory $T$ are at a distance at most $y$ from $v$. Then, if agents started at two nodes at a distance $3y$, their trajectories would never intersect and hence rendezvous would be impossible, which is a contradiction.

Suppose that there exists an infinite sequence $(d_0,d_1,d_2,\dots)$ of positive integers, diverging to infinity,  such that there exist nodes 
$(v_0,v_1,v_2,\dots)$ of the grid, with $v_i$ at a distance $d_i$ from $v$, and none of the nodes $v_i$ belongs to $T$. For any $i$, consider the first round $\delta_i$ in which the agent starting at $v$ in round 0 reaches a node at a distance at least $4d_i^{\sqrt{2}}$ from $v$. The adversary puts the other agent at node $v_i$ and wakes it up in round $\delta_i$. Hence, in round $\delta_i$, the two agents are at a distance at least
$4d_i^{\sqrt{2}}-d_i \geq 3d_i^{\sqrt{2}}$ from each other. Counting time starts in round $\delta_i$. Since the agents can decrease their distance by at most 2 in each round, it follows that rendezvous time is larger than $d_i^{\sqrt{2}}$. On the other hand, the initial distance between the agents is $d_i$, which proves the theorem in this case, with constant $c=1$.

Hence, in the rest of the proof we may suppose that, for some $d>0$, all nodes at a distance larger than $d$ from $v$ are in the trajectory $T$.
For any natural number $x$, let $R(v,x)$ denote the set of nodes of the grid whose distances from $v$ are larger than $x$ but at most $2x$. Hence, for any $x\geq d$, the number of nodes in $R(v,x)$ is at least $2x^2$ and all nodes of $R(v,x)$ are in the trajectory $T$.
Let $X$ be the set of integers larger than $d$.
Define the following function $f: X \longrightarrow \mathcal{N}$, where $ \mathcal{N}$ is the set of natural numbers. Fix an integer $x\in X$. Among all the distances between $v$ and all the nodes that are in the prefix of $T$ at the time when $T$ visits the last node of $R(v,x)$ define $f(x)$ to be the largest one.
 By definition, for any $x \in X$, we have $f(x) \geq 2x$.  There can be many nodes at distance $f(x)$ from $v$ in this prefix. Call the first of them the {\em witness} of $x$. Consider two complementary cases.

\noindent
{\bf Case 1.} There exists an infinite sequence $(x_0,x_1,x_2,\dots)$ of positive integers larger than $d$, diverging to infinity, such that $f(x_i)>3x_i^{\sqrt{2}}$.

Fix a natural number $i$ and let $w$ be the witness of $x_i$. Consider an agent starting at node $v$. Consider the round $t_i$ when the trajectory $T$ visits the last node $u$ belonging to $R(v,x_i)$. Hence, in some round $\delta _i < t_i$ the trajectory visited node $w$. The adversary puts the other agent at node $u$ and wakes it up in round  $\delta _i$. Hence, in round $\delta_i$, the two agents are at a distance at least $f(x_i)-2x_i>3x_i^{\sqrt{2}}-2x_i>x_i^{\sqrt{2}}$. Counting time starts in round $\delta_i$. 
The initial distance between the agents is $d_i$, where $x_i<d_i\leq 2x_i $. The sequence $(d_0,d_1,d_2,\dots)$ diverges to infinity.
Since the agents can decrease their distance by at most 2 in each round, it follows that rendezvous time is larger than $x_i^{\sqrt{2}}/2\geq d_i^{\sqrt{2}}/(2\cdot 2^{\sqrt{2}} )$. This proves the theorem in this case, with constant $c=1/(2\cdot 2^{\sqrt{2}} ) $.

\noindent
{\bf Case 2.} For sufficiently large integers $x$, we have $f(x)\leq 3x^{\sqrt{2}}$. 

Let $g>d$ be a constant such that $f(x)\leq 3x^{\sqrt{2}}$ for all $x \geq g$. For any natural number $i$, let $x_i=g+i$ and let
$d_i=3f(x_i)$. $(d_0,d_1,d_2,\dots)$ is an infinite sequence of positive integers at least $g$, diverging to infinity.
Fix a natural number $i$. Let $u'$ and $u''$ be two nodes at distance $d_i$.
The adversary puts an agent $a'$ at node $u'$, an agent $a''$ at node $u''$ and wakes them up simultaneously in round 0 (i.e., $\delta_i=0$).
Let $t'_i$ be the round in which agent $a'$ visits the the last node of $R(u',x_i)$ and let $t''_i$ be the round in which agent $a''$ visits the last node of  $R(u'',x_i)$. Before round $t'_i$ agent $a'$ can be at a distance at most $f(x_i)$ from $u'$ and before round $t''_i$ agent $a''$ can be at a distance at most $f(x_i)$ from $u''$. Since $d_i=3f(x_i)$,  their respective trajectories until these rounds are disjoint and $t'_i=t''_i$.
Call this common round $t_i$. Thus agents cannot meet before round $t_i$. By the definition of $t_i$ and in view of the fact that each of the sets
$R(u',x_i)$ and $R(u'',x_i)$ has at least $2x_i^2$ nodes, we have $t_i\geq 2x_i^2-1\geq x_i^2$.
Since $f(x_i)\leq 3 x_i^{\sqrt{2}}$, we have $x_i \geq (f(x_i)/3)^{1/\sqrt{2}}$. Hence we get
$t_i\geq  x_i^2 \geq (f(x_i)/3)^{\sqrt{2}}=d_i^{\sqrt{2}}/9^{\sqrt{2}}$. This proves the theorem in this case, with constant $c=1/9^{\sqrt{2}} $.
It follows that the theorem is proven with constant $c=1/23<1/9^{\sqrt{2}}<1/(2\cdot 2^{\sqrt{2}} ) <1$.
\end{proof}

The above theorem implies the following corollary:

\begin{corollary}
Suppose that the agents can be placed by the adversary in arbitrary nodes of the grid at a distance at most $D$ unknown to them and woken up
with arbitrary delay.
No deterministic rendezvous algorithm in this scenario can have time complexity
$o(D^{\sqrt{2}})$.
\end{corollary}

On the positive side, we show the rendezvous algorithm {\tt Hardest Scenario} working in time $O(D^2)$, where agents can be placed by the adversary at arbitrary nodes of the grid at a distance at most $D$, unknown to them, and woken up
with arbitrary delay. The algorithm is presented in Section \ref{sect-alg-no-assump}, and 
 the proof of its correctness and complexity is deferred to Section  \ref{sect-proof-no-assump}. 
 
 \subsection{The Algorithm}
\label{sect-alg-no-assump}

The Algorithm {\tt Hardest Scenario} consists of three parts. Whenever rendezvous happens, the algorithm is interrupted immediately.
In the first part, the agent performs the following procedure that is interrupted when the agent makes the first hit.
Intuitively, the procedure produces an infinite spiral whose consecutive sides are in the four cardinal directions $N,W,S,E$, and the spiral fills the entire grid.\\

\noindent
{\bf Procedure} {\tt Squarespiral}\\

\noindent
$p:=1$;\\
{\bf repeat forever}\\
\hspace*{1cm}{\tt go} $N$ for $2p-1$ rounds\\ 
\hspace*{1cm}{\tt go} $W$ for $2p-1$ rounds\\ 
\hspace*{1cm}{\tt go} $S$ for $2p$ rounds\\ 
\hspace*{1cm}{\tt go} $E$ for $2p$ rounds\\ 
\hspace*{1cm}$p:=p+1$;\\

Suppose that the first hit is made at node $u$. (Since the spiral fills the grid, a hit is made at the latest at the base of the other agent).
Let $T$ be the trajectory of the agent traversed in the first part of the algorithm, from the base to node $u$. Let $\hat{T}$ be the reverse trajectory, from $u$ to the base $b$ of the agent.  
In the sequel, we call $T$ the main trajectory of the agent.
Consider all vertical and horizontal lines of the grid that intersect $T$.
We call the distance between the leftmost and the rightmost vertical lines the {\em width} of $T$ and denote it by $T.width$; similarly, we call the distance between the top and the bottom horizontal lines the {\em height} of $T$ and denote it by $T.height$.
We denote by $pred_T(u)$ the predecessor of $u$ in the trajectory $T$.

At the time of the first hit, the second part starts.
The agent executes the following procedure.\\

	\noindent
{\bf Procedure} {\tt PredecessorProbe($T, u$)}\\
\hspace*{1cm} {\bf if} the agent made a horizontal hit at $u$ {\bf then}\\
\hspace*{2cm}stay at $u$ for $T.height$ rounds;\\
\hspace*{1cm} {\bf else} \\
\hspace*{2cm}stay at $u$ for $T.width$ rounds;\\
\hspace*{1cm}move to $pred_T(u)$;\\
\hspace*{1cm}move to $u$;\\

After completing the above procedure, the agent is back at the node $u$ where it made the first hit, and the third part starts.
The agent chooses one of the two following actions: 
\begin{itemize}
	\item Action $I$: perform the following procedure {\tt Swing}, interrupting it at the rendezvous:\\
	
	\noindent
	{\bf Procedure} {\tt Swing}\\
	{\bf repeat forever}\\
	\hspace*{1cm}follow $\hat{T}$; follow $T$;\\
	
	Intuitively, the agent ``swings'' on the part of the spiral produced in the first part of the algorithm, from node $u$ to the base of the agent and back, until the other agent is met.
	
	\item Action $II$: stay at $u$ forever.	
\end{itemize}

It remains to establish the rule of choosing the above actions. In order to do that,
we define, for any node $s$, its $Dir$-neighbor to be the neighbor of $s$ corresponding to port $Dir$ at $s$, where $Dir \in \{N, S, W, E\}$.
The decision which of these actions has to be chosen is given in Table \ref{tab-decision-unknown}.
It is based on which of the four neighbors of $u$ are on the trajectory $T$ of the agent.

\begin{table}[!h]
	\begin{center}
		\begin{tabular}{ |c|c c c c |c |} 
			\hline
			\multirow{ 2}{*}{Row}&\multicolumn{4}{|c|}{Input} & \multirow{ 2}{*}{Action} \\
			\cline{2-5}
			&$W$-neighbor & $S$-neighbor & $E$-neighbor & $N$-neighbor & \\
			\hline \hline
			1&0 &1 & 0 or 1 & 0  & II \\
			\hline
			2&0 or 1 & 0 & 0 & 1  & I   \\
			\hline
			3&0 & 0 & 1 & 0 or 1  & I   \\
			\hline
			4&1 & 0 or 1 & 0 & 0  & II   \\
			\hline
		\end{tabular}
		\caption{\label{tab-decision-unknown}The decision table for Algorithm {\tt Hardest Scenario}. }
	\end{center}
\end{table}
In the table, each row records two of the eight possible inputs the agent might see, and the corresponding action that the agent should take.
More specifically, if the $Dir$-neighbor of $u$ is on the trajectory of the agent, then we  write ``1'' in the input column that corresponds to the $Dir$-neighbor; otherwise, we write ``0''.
Take the first row as an example. This row corresponds to the following two inputs: ``the only neighbor of $u$ on $T$ is the $S$-neighbor'' and ``the only neighbors of $u$ on $T$ are the $S$-neighbor and the $E$-neighbor''.
For each of these inputs, the agent chooses action II.
Note that all inputs listed in the table are pairwise exclusive, hence the choice of the action is unambiguous.
We will prove that the eight inputs noted in the table are the only ones that can happen.

\subsection{Correctness and complexity}\label{sect-proof-no-assump}

The high-level idea of the proof of correctness is to show that, either rendezvous occurs before each agent starts executing the third part of Algorithm {\tt Hardest Scenario}, or both agents get, in the beginning of the third part of the algorithm, some of the inputs enumerated in Table \ref{tab-decision-unknown}, and choose  different actions based on these inputs.
This is enough to guarantee rendezvous in the end. Indeed, 
let $a_I$ and $a_{II}$ be the agents that perform actions $I$ and $II$, respectively.
After performing action $II$, agent $a_{II}$ stays at a node $u$ in the main trajectory of $a_I$. 
On the other hand, $a_{I}$, which chooses action $I$, keeps performing Procedure {\tt Swing} until rendezvous.
After $a_{II}$ stabilizes at $u$, at most one execution of Procedure {\tt Swing} performed by agent $a_{I}$ guarantees that rendezvous occurs.
Before proving the correctness of Algorithm {\tt Hardest Scenario}, we introduce some notations and a fact relevant to the algorithm, as well as Lemma \ref{lem-hitting-region}.

Let $B$ and $R$ denote the two agents. 
If one agent is denoted by $X$, then we use $\overline{X}$ to denote the other agent. Hence $\overline{B}=R$ and $\overline{R}=B$.
We assume that agent $X$, where $X\in \{B, R\}$, makes its first hit at node $u_X$ in round $t_X$.
Node $u_X$ is always different from node $u_{\overline{X}}$ 
and each agent makes at most one hit, in view of Algorithm {\tt Hardest Scenario}.
Let $T_X$ denote the the main trajectory of agent $X$, where $X\in\{B, R\}$.
The number of edges on $T_X$ is denoted by $|T_X|$.
Without loss of generality we may assume that  $|T_X|>0$.

\begin{fact}
	\label{fact-counterclockwise}
	Since each agent follows its square spiral trajectory counterclockwise in the first part of the algorithm, all the following twelve statements, concerning moves of the agent in this part, are straightforward.
	\begin{itemize}
		\item[] i) If an agent is going North in round $t$ and arrives at a node $v$ in round $t$, then no nodes that are right of $v$ can be on the trajectory of this agent by round $t$.
		\item[] ii) If an agent is going East in round $t$ and arrives at a node $v$ in round $t$, then no nodes that are below $v$ can be on the trajectory of this agent by round $t$.
		\item[] iii) If an agent is going South in round $t$ and arrives at a node $v$ in round $t$, then no nodes that are left of $v$ can be on the trajectory of this agent by round $t$.
		\item[] iv) If an agent is going West in round $t$ and arrives at a node $v$ in round $t$, then no nodes that are above $v$ can be on the trajectory of this agent by round $t$.
		\item[] v) If an agent visits a node $v$ for the first time before visiting the $W$-neighbor of $v$ then, in the round when the agent visits the $W$-neighbor of $v$ for the first time, the agent is moving either South or West.
		\item[] vi) If an agent visits a node $v$ for the first time before visiting the $N$-neighbor of $v$ then, in the round when the agent visits the $N$-neighbor of $v$ for the first time, the agent is moving either West or North.
		\item[] vii) If an agent visits a node $v$ for the first time before visiting the $E$-neighbor of $v$ then, in the round when the agent visits the $E$-neighbor of $v$ for the first time, the agent is moving either North or East.
		\item[] viii) If an agent visits a node $v$ for the first time before visiting the $S$-neighbor of $v$ then, in the round when the agent visits the $S$-neighbor of $v$ for the first time, the agent is moving either East or South.
		\item [] ix) If in round $t_X$, i.e., the round when agent $X$ makes its first hit, where $X\in \{B, R\}$, the agent $X$ is moving East, then $T_X.height$ is an even number.
		\item [] x) If in round $t_X$, agent $X$ is moving North, then $T_X.width$ is an even number.
		\item [] xi) If in round $t_X$, agent $X$ is moving West, then $T_X.height$ is an odd number.
		\item [] xii) If in round $t_X$, agent $X$ is moving South, then $T_X.width$ is an odd number.
	\end{itemize}
\end{fact}

Define a {\em corner} as any node where agent $X$, for $X\in\{B, R\}$, changes direction when it follows its main trajectory $T_X$.
Let $a,b,c,d$ be the first four corners on the reverse trajectory $\hat{T}_X$.
If the trajectory has fewer than 4 corners, then we assign all unused labels from the set $\{a, b,c,d\}$ to the base of the agent.
For example, if the trajectory $\hat{T}_X$ has only two corners, then we label them as $a$ and $b$, respectively, and label the base of the agent as both $c$ and $d$.
Consider the set of all neighbors of $u_X$ on $T_X$. 
Let $v_X$ denote the node in this set closest to $d$.
Note that node $v_X$ is always unique.
We define the {\em late part} of the trajectory $T_X$ as follows.
If $v_X=pred_{T_X}(u_X)$ then the late part of $T_X$ is its suffix between $d$ and $u_X$; otherwise, it is its suffix between $v_X$ and $u_X$ (cf. Fig. \ref{fig-late-region}).

\begin{figure}[h!]
	\centering
	\begin{subfigure}[t]{0.49\textwidth}
		\centering
			\includegraphics[scale=1]{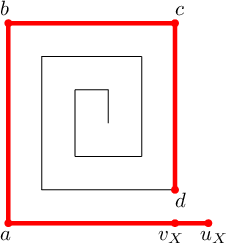}
		\caption{}
	\end{subfigure}%
	\begin{subfigure}[t]{0.49\textwidth}
		\centering
		\includegraphics[scale=1]{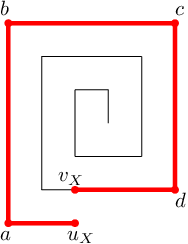}
		\caption{}
	\end{subfigure}
	\caption{\label{fig-late-region}Late parts of trajectories. In figure (a), we have $v_X=pred_{T_X}(u_X)$. In figure (b), we have  $v_X\neq pred_{T_X}(u_X)$. The late part of each trajectory is highlighted in red and bold.}
\end{figure}


\begin{lemma}
	\label{lem-hitting-region}
	If only one neighbor of $u_X$ is on $T_X$ then the node where agent $\overline{X}$ could make a hit is either $pred_{T_X}(u_X)$ or a node in the segment between $c$ and $d$.
	If exactly two neighbors of $u_X$ are on $T_X$, then the node where $\overline{X}$ could make a hit is either $pred_{T_X}(u_X)$ or a node in the segment between $d$ and $v_X$.
\end{lemma}
\begin{proof}
	We start with the following claim.
	
	\begin{claim}
		\label{claim-hitting-exlusive}
		Let $X\in \{B, R\}$. Then any node at which $\overline{X}$ can make a hit is on the late part of the trajectory $T_X$.

	\end{claim}
	
	In order to prove the claim, let $s$ denote any node on $T_X$ but not on the late part of $T_X$ and let $s'$ denote any node on the late part of $T_X$.
	Observe that the difference between nodes $s$ and $s'$ is that all the neighbors of $s$ are on $T_X$ in round $t_X-1$, while at least one neighbor of $s'$ is not on $T_X$ in this round.
 We show that agent $\overline{X}$ can never visit node $s$.
 
 Notice that node $s$ cannot be the base of agent $\overline{X}$.
 Indeed, let $t$ denote the round when agent $X$ visits $s$ for the first time.
 Then $t<t_X$, as node $s$ is different from $u_X$ and node $s$ is visited by agent $X$ earlier than $u_X$ is.
 If node $s$ was the base of agent $\overline{X}$, then agent $X$ would make its first hit no later than in round $t$.
 This contradicts the fact that agent $X$ makes its first hit in round $t_X$.
 So, $s$ cannot be the base of agent $\overline{X}$.
 
 Hence, agent $\overline{X}$ cannot visit $s$, without visiting at least one neighbor of $s$.
 Note that all the neighbors of $s$ are domestic for agent $X$, since agent $X$ has visited all of them before round $t_X$ and no hit was made by $X$ at any of them.
  When agent $\overline{X}$ arrives at some neighbor of $s$, it makes a hit there and never explores any nodes it has not visited previously.
 Therefore, agent $\overline{X}$ can never visit $s$.
 This proves that ${\overline{X}}$ can make a hit only on the late part of $T_X$. $\diamond$
	
	
	We will call the nodes in the late part of $T_X$ the {\em potential nodes} for $\overline{X}$.
	 	In view of Claim \ref{claim-hitting-exlusive}, 
	only the potential nodes are the ones where agent $\overline{X}$ could make a hit.

	In round $t_X$, agent $X$ could be moving in one of the four cardinal directions $N,W,S$ or $E$.
	We give the proof for the case in which $X$ moves East in round $t_X$, and the proofs for the other three cases are similar.
	In view of Algorithm {\tt Hardest Scenario}, we know that the $W$-neighbor of $u_X$ is always on $T_X$, but neither the $E$-neighbor nor the $S$-neighbor of $u_X$ is.
	If the $N$-neighbor of $u$ is not on $T_X$, then only one neighbor of $u_X$ is on $T_X$.
	Otherwise, exactly two neighbors of $u_X$ are on $T_X$.
	We consider both cases separately.
	
	First, assume that the $N$-neighbor of $u_X$ is not on $T_X$ (cf. Fig. \ref{fig-late-region} (a)).
	Note that in this case, nodes $pred_{T_X}(u_X)$, the $W$-neighbor of $u_X$ and $v_X$ are all equal.
	It follows that node $pred_{T_X}(u_X)$ cannot be on the trajectory $T_{\overline{X}}$, by round $t_X-1$; otherwise, a hit would be made by agent $X$ in round $t_X-1$.
	Since agent $\overline{X}$ has visited $u_X$ by round $t_X$ and agent $\overline{X}$ cannot visit $u_X$ and $pred_{T_X}(u_X)$ in the same round, we know that $u_X$ has to be visited by $\overline{X}$ for the first time before its $W$-neighbor, $pred_{T_X}(u_X)$, is.
	In view of part $v)$ of Fact \ref{fact-counterclockwise},
	if $\overline{X}$ visits $pred_{T_X}(u_X)$, then in the round when $\overline{X}$ visits $pred_{T_X}(u_X)$ for the first time, $\overline{X}$ moves West or South.
	If $\overline{X}$ moves West, it makes a hit at $pred_{T_X}(u_X)$, since $\overline{X}$ cannot visit $pred_{T_X}(u_X)$ for the first time by round $t_X-1$, and agent $X$ is exactly at $pred_{T_X}(u_X)$ in round $t_X-1$.
	Therefore, $pred_{T_X}(u_X)$ is a node where agent $\overline{X}$ might make a hit.
	If $\overline{X}$ moves South, it could make a hit before visiting $pred_{T_X}(u_X)$.
	Regardless of whether $\overline{X}$ makes its first hit at $pred_{T_X}(u_X)$ or before visiting $pred_{T_X}(u_X)$, in view of part $iii)$ of Fact \ref{fact-counterclockwise}, agent $\overline{X}$ cannot visit any nodes that are left of $pred_{T_X}(u_X)$.
	Therefore, $\overline{X}$ makes a hit either at $pred_{T_X}(u_X)$ or at a node in the segment between $c$ and $d$.
	
	Next, assume that the $N$-neighbor of $u_X$ is on $T_X$ (cf. Fig. \ref{fig-late-region} (b)).
	Note that in this case, node $v_X$ is the $N$-neighbor of $u_X$.
	Using the same arguments as in the previous paragraph, agent $\overline{X}$ cannot visit any nodes that are left of $pred_{T_X}(u_X)$.
	So, all potential nodes that are left of $pred_{T_X}(u_X)$ are excluded.
	Observe that node $v_X$ is domestic for $X$ and node $u_X$ is domestic for $\overline{X}$.
	It follows that agent $\overline{X}$ has to visit $u_X$ for the first time before visiting the $N$-neighbor $v_X$ of $u_X$.
	Otherwise, agent $\overline{X}$ would make a hit at $v_X$, and would never visit $u_X$.
	In view of part $vi)$ of Fact \ref{fact-counterclockwise}, if agent $\overline{X}$ visits $v_X$, then $\overline{X}$ moves either North or West in the round when it visits $v_X$ for the first time.
	If agent $\overline{X}$ moves North, it makes a hit at node $v_X$, since $v_X$ is domestic for agent $X$.
	If agent $\overline{X}$ moves West, it could make a hit before visiting $v_X$.
	Regardless of whether $\overline{X}$ makes its first hit at $v_X$ or before visiting $v_X$, it cannot visit any node that is above $v_X$, in view of part $iv)$ of Fact \ref{fact-counterclockwise}.
	Therefore, agent $\overline{X}$ makes its first hit either at $pred_{T_X}(u_X)$ or at a node in the segment between $d$ and $v_X$.
	This completes the proof.
\end{proof}

The rest of the proof of correctness of the algorithm is divided into two cases: either $u_{{X}}= pred_{T_{\overline{X}}}(u_{\overline{X}})$ for some $X\in\{B, R\}$ or $u_{{X}}\neq pred_{T_{\overline{X}}}(u_{\overline{X}})$ and $u_{\overline{X}} \neq pred_{T_{{X}}}(u_{{X}})$.

\noindent

\begin{lemma}
	\label{theorem-ends-in-2nd-part}
	Let $t_{X, 3}$ denote the round when agent $X$, for any $X\in\{B, R\}$, starts executing the third part of Algorithm {\tt Hardest Scenario}.
	If $u_{{X}}= pred_{T_{\overline{X}}}(u_{\overline{X}})$ for some $X\in\{B, R\}$, then either rendezvous happens before round $\min(t_{X, 3}, t_{\overline{X}, 3})$ or agents have chosen different actions, from Table \ref{tab-decision-unknown}, by round $\max(t_{X, 3}, t_{\overline{X}, 3})$.
\end{lemma}

\begin{proof}
We consider two cases: either $|T_{\overline{X}}|>1$ or $|T_{\overline{X}}|= 1$.

\noindent
{\bf Case A.  $|T_{\overline{X}}|> 1$.}

For any $Dir \in \{N, W, S, E\}$, let $\hat{Dir}$ denote the opposite direction of $Dir$.
For instance, if $Dir=N$, then $\hat{Dir}=S$.
We denote by $CWT(Dir)$ (resp. $CCWT(Dir)$) the direction obtained by rotating $Dir$ by $\pi/2$  clockwise (resp. counterclockwise).
For example, $CWT(N)=E$ and $CCWT(N)=W$.
Hence, $CWT(CWT(Dir))=CCWT(CCWT(Dir))=\hat{Dir}$.
We start with the following claim.

\begin{claim}
	\label{claim-direction-overline-X}
	If agent $X$, where $X\in\{B, R\}$, moves in direction $Dir$ in round $t_X$ and agent ${X}$ makes a hit at $pred_{T_{\overline{X}}}(u_{\overline{X}})$, then agent $\overline{X}$ moves in direction $\hat{Dir}$ in round $t_{\overline{X}}$.
\end{claim}

In order to prove the claim, notice that,
in round $t_{\overline{X}}$, agent $\overline{X}$ moves from $pred_{T_{\overline{X}}}(u_{\overline{X}})$ (which is equal to $u_X$) to a node $v$ that is one of the four neighbors of $u_X$, and makes a hit at node $v$.
If agent $\overline{X}$ arrived at $v$ by a $Dir$-move, then $v$ would not be visited by agent $X$, and agent $\overline{X}$ could not make a hit at node $v$.
Therefore, agent $\overline{X}$ cannot move in direction $Dir$ in round $t_{\overline{X}}$.
Assume that agent $\overline{X}$ moves in direction $CCWT(Dir)$ and visits $v$ in round $t_{\overline{X}}$.
In view of Procedure {\tt Squarespiral}, the $CCWT(CCWT(Dir))$-neighbor of $pred_{T_{\overline{X}}}(u_{\overline{X}})$ must be on $T_{\overline{X}}$, and must be domestic for agent $\overline{X}$, since $|T_{\overline{X}}|>1$.
The $CCWT(CCWT(Dir))$-neighbor of $pred_{T_{\overline{X}}}(u_{\overline{X}})$ is $pred_{T_{{X}}}(u_{{X}})$.
As we know, $pred_{T_{{X}}}(u_{{X}})$ is domestic for $X$.
This is a contradiction which proves that agent $\overline{X}$ cannot move in direction $CCWT(Dir)$ in round $t_{\overline{X}}$.
If agent $\overline{X}$ moved in direction $CWT(Dir)$ and visited $v$ in round $t_{\overline{X}}$, then $v$ would not be domestic for agent $X$, in view of parts $i)-iv)$ of Fact \ref{fact-counterclockwise}.
As a result, agent $\overline{X}$ could not make a hit at node $v$.
Thus, agent $\overline{X}$ cannot move in direction $CWT(Dir)$ in round $t_{\overline{X}}$.
It follows that agent $\overline{X}$ can only move in direction $\hat{Dir}$ in round $t_{\overline{X}}$. $\diamond$

If agent ${X}$ moves in direction $Dir$ in round $t_X$ and makes a hit at $pred_{T_{\overline{X}}}(u_{\overline{X}})$, i.e., $u_{{X}}=pred_{T_{\overline{X}}}(u_{\overline{X}})$, then agent $\overline{X}$ moves in direction $\hat{Dir}$ in round $t_{\overline{X}}$, from $u_{{X}}$ to $pred_{T_{{X}}}(u_{{X}})$, in view of Claim \ref{claim-direction-overline-X}.
Thus, we have $u_{\overline{X}}=pred_{T_X}(u_X)$ and $u_X=pred_{T_{\overline{X}}}(u_{\overline{X}})$.
Next, we show that both hits are made in the same round.

\begin{claim}
	\label{claim-pred-pred}
	If agent $X$ makes its first hit at $pred_{T_{\overline{X}}}(u_{\overline{X}})$ and agent $\overline{X}$ makes its first hit at $pred_{T_{{X}}}(u_{{X}})$, then both hits are made in the same round.
\end{claim}

	In order to prove the claim, notice that agent $X$ visits $pred_T(u_X)$ in round $t_X-1$ and agent $\overline{X}$ visits $pred_{\overline{T}}(u_{\overline{X}})$ in round $t_{\overline{X}}-1$.
	If agent $X$ makes a hit at $pred_{T_{\overline{X}}}(u_{\overline{X}})$ and agent $\overline{X}$ makes a hit at $pred_{T_{{X}}}(u_{{X}})$, then it follows that $u_X=pred_{\overline{T}}(u_{\overline{X}})$ and  $u_{\overline{X}}=pred_T(u_X)$.
	Since $u_X$ is domestic for $\overline{X}$, $\overline{X}$ arrives at $u_X$ for the first time before $X$ does, and thus we have $t_{\overline{X}}-1\le t_{X}$.
	Symmetrically, $u_{\overline{X}}$ is domestic for $X$, which implies that $X$ arrives at $u_{\overline{X}}$ for the first time before $\overline{X}$ does.
	So, we have $t_X-1\le t_{\overline{X}}$.
	Overall, we have $t_X-1\le t_{\overline{X}}\le t_X+1$.
	If $t_X-1=t_{\overline{X}}$, then both agents $X$ and $\overline{X}$ would meet at node $u_{\overline{X}}$ in round $t_{\overline{X}}$, and the algorithm would end in its first part.
	Similarly, if $t_{\overline{X}}-1=t_X$, then both agents would meet at node $u_X$ in round $t_X$,
	and the algorithm would end in its first part.
	Hence, we may assume that $t_X-1< t_{\overline{X}}< t_X+1$, and thus $t_X=t_{\overline{X}}$.
	This proves the claim. $\diamond$
	
		In view of Claim \ref{claim-pred-pred}, both agents start the second part of the algorithm in the same round $t_X=t_{\overline{X}}$.
		In view of Claim \ref{claim-direction-overline-X}, both agents move in opposite directions in round $t_X$.
		If both agents move horizontally in round $t_X$, i.e., one moves East and the other one moves West, then agent $X$ stays at $u_{X}$ for $T_X.height$ rounds and agent $\overline{X}$ stays at $u_{\overline{X}}$ for $T_{\overline{X}}.height$ rounds, in view of Procedure {\tt PredecessorProbe}.
		Note that $T_X.height$ is different from $T_{\overline{X}}.height$, since one of them is even and the other one is odd, in view of parts $ix)$ and $xi)$ of Fact \ref{fact-counterclockwise}.
		Therefore, both agents that started the second part simultaneously and waited for different periods of time, start moving in different rounds in the second part of Algorithm {\tt Hardest Scenario}.
		The agent that moves first finds the other agent after one round, since both agents stay at nodes joined by an edge and move on this edge.
		In the case when both agents move vertically in round $t_X=t_{\overline{X}}$, the proof is similar.
		Therefore, rendezvous is achieved in the second part of the algorithm.
		This completes the proof for Case A.

\noindent
{\bf Case B.  $|T_{\overline{X}}|= 1$.}

Since $|T_{\overline{X}}|= 1$, trajectory $T_{\overline{X}}$  contains only two nodes: the base of agent $\overline{X}$ and its North neighbor, denoted by $z$.
	Note that if both agents visited node $z$ for the first time in the same round, then rendezvous would happen and Algorithm {\tt Hardest Scenario} would be interrupted immediately.
	Hence, neither $\overline{X}$ would make a hit nor $T_{\overline{X}}$ would be defined.
	Thus we may assume that $z=u_{\overline{X}}$ and $z$ is domestic for agent $X$. Moreover, 
	the base of agent $\overline{X}$ is $pred_{T_{\overline{X}}}(u_{\overline{X}})$ and $pred_{T_{\overline{X}}}(u_{\overline{X}})=u_X$.
	We now establish the direction in which agent $X$ may move in round $t_X$.
	
	\begin{claim}
		\label{claim-direction-of-X}
Agent $X$ moves either South or East in round $t_X$.
	\end{claim}
	
In order to prove the claim, we show that agent $X$ can neither move North nor West in round $t_X$.
Since agent $X$ makes a hit at $pred_{T_{\overline{X}}}(u_{\overline{X}})$ in round $t_X$, if agent $X$ moved North in round $t_X$, then all nodes that are above $pred_{T_{\overline{X}}}(u_{\overline{X}})$ would not be on the trajectory $T_X$, including the North neighbor of $pred_{T_{\overline{X}}}(u_{\overline{X}})$, which is $u_{\overline{X}}$.
This contradicts the fact that $u_{\overline{X}}$ is domestic for $X$, so agent $X$ cannot move North in round $t_X$.
If agent $X$ moved West in round $t_X$, then $u_{\overline{X}}$ would not be on $T_X$, in view of part $iv)$ of Fact \ref{fact-counterclockwise}.
Therefore, agent $X$ cannot move West in round $t_X$, either.
This proves the claim. $\diamond$.

In view of Claim \ref{claim-direction-of-X}, agent $X$ moves either South or East in round $t_X$.
If it moves South, then $u_X =pred_{T_{\overline{X}}}(u_{\overline{X}})$ and $u_{\overline{X}}=pred_{T_{{X}}}(u_{{X}})$.
Since $|T_{\overline{X}}|= 1$, we must have $|T_X|>1$. Hence, in view of Case A applied to agent $X$ (instead of agent $\overline{X}$)  we conclude that rendezvous is achieved in the second part of the algorithm because $u_{\overline{X}}=pred_{T_{{X}}}(u_{{X}})$.

Next, we consider the case when agent $X$ moves East in round $t_X$.
After making a hit in round $t_{\overline{X}}$, agent $\overline{X}$ starts executing the second part of the algorithm.
	In view of Procedure {\tt PredecessorProbe}, agent $\overline{X}$ moves to $pred_{T_{\overline{X}}}(u_{\overline{X}})$ in round $t_{\overline{X}}+1$, since $T_{\overline{X}}$.width is $0$.
	If it does not meet agent $X$ there, then it moves back to $u_{\overline{X}}$ in round $t_{\overline{X}}+2$.
	Then it starts the execution of the third part of the algorithm.
	Thus, we have $t_{\overline{X}, 3}=t_{\overline{X}}+3$.
	As only the $S$-neighbor of $u_{\overline{X}}$ is on $T_{\overline{X}}$ and the other three neighbors are not, agent $\overline{X}$ chooses action $II$ (in view of Table \ref{tab-decision-unknown}). 
	
	Observe that $t_X\ge t_{\overline{X}}$; otherwise, rendezvous would have happened even before agent $\overline{X}$ was woken up by the adversary.
	As agent $X$ moves East in round $t_X$, it follows that $T_X.height\ge 2$.
	Agent $X$ stays at $u_X$ (which is $pred_{T_{\overline{X}}}(u_{\overline{X}})$) from round $t_X+1$ to round $t_X+T_X.height$, in view of Procedure {\tt PredecessorProbe}.
	If $t_X=t_{\overline{X}}$, then agent $X$ is found by agent $\overline{X}$ in round $t_{\overline{X}}+1$, while staying at $u_X$, which means that rendezvous occurs in round $t_{\overline{X}}+1$.
	Otherwise, agent $X$ moves to node $pred_{T_X}(u_X)$ in round $t_X+T_X.height+1$ and moves back to $u_X$ in round $t_X+T_X.height+2$, in view of Procedure {\tt PredecessorProbe}.
	Then, it starts executing the third part of the algorithm.
	Thus, we have $t_{X, 3}=t_X+T_X.height+3$, which is larger than $t_{\overline{X}, 3}$.
	Agent $X$ chooses action $I$, in view of Table \ref{tab-decision-unknown}, as the only neighbors of $u_X$ on $T_X$ are the $W$-neighbor and the $N$-neighbor.
	Therefore, if agent $X$ moves East in round $t_X$, then either rendezvous occurs in a round before $t_{\overline{X}, 3}$ or agents have chosen different actions by round $t_{{X}, 3}$.
	This completes the proof for Case B.
\end{proof}

In Lemma \ref{theorem-ends-in-2nd-part}, we showed that, if $u_{{X}}= pred_{T_{\overline{X}}}(u_{\overline{X}})$, for some $X\in\{B, R\}$, then either rendezvous is achieved, or agents choose different actions. In the next theorem, we consider the complementary condition, which is $u_{{X}}\neq pred_{T_{\overline{X}}}(u_{\overline{X}})$ and $u_{\overline{X}} \neq pred_{T_{{X}}}(u_{{X}})$.

\begin{lemma}
	\label{theorem-no-pred-diff-actions}
If $u_{{X}}\neq pred_{T_{\overline{X}}}(u_{\overline{X}})$ and $u_{\overline{X}} \neq pred_{T_{{X}}}(u_{{X}})$, then both agents choose different actions in the third part of Algorithm {\tt Hardest Scenario}.
\end{lemma}
\begin{proof}
First, we prove that the only possible inputs obtained by an agent at the end of the second part of Algorithm {\tt Hardest Scenario} are those  given in Table \ref{tab-decision-unknown}.
\begin{claim}
	\label{claim-num-neighbors}
	Let $\ell_X$, for  $X\in\{B, R\}$, be the number of neighbors of $u_X$ that are on the trajectory $T_X$. Then $1\leq \ell_X \leq 2$.
\end{claim}

In order to prove the claim, notice that $pred_{T_X}(u_X)$ is always on $T_X$, 
since $|T_X|>0$.
Therefore, $\ell_X\ge 1$.
Let $Dir$ denote the direction in which agent $X$ moves in round $t_X$.
In view of parts $i)-iv)$ of Fact \ref{fact-counterclockwise}, the $CWT(Dir)$-neighbor of $u_X$ is not on $T_X$
(where $CWT(Dir)$ denotes the direction obtained by rotating $Dir$ by $\pi/2$  clockwise).
Furthermore, the $Dir$-neighbor of $u_X$ is never on $T_X$. since $u_X$ has four neighbors, we have $\ell_X\le 2$.
This proves the claim. $\diamond$

If $u_X$ has exactly two neighbors on the trajectory $T_X$, it follows that both neighbors cannot be on the same line, in view of Procedure {\tt Squarespiral}.
In view of Claim \ref{claim-num-neighbors}, Table \ref{tab-decision-unknown} enumerates all the possible inputs that agent $X$ might see.
Recall that each row of Table \ref{tab-decision-unknown} corresponds to two different inputs.
For any $1 \leq j \leq 4$, define  $inputs (j)$ to be the pair of inputs listed in row $j$ of the table.
Using the assumption $u_{\overline{X}} \neq pred_{T_{{X}}}(u_{{X}})$, we can correctly infer the relative positions between $u_X$ and $u_{\overline{X}}$ from these inputs. The next claim follows from Lemma \ref{lem-hitting-region}.

\begin{claim}
	\label{claim-relative-positions-of-hitting}
	If $u_{\overline{X}} \neq pred_{T_{{X}}}(u_{{X}})$, then
	the relative positions between $u_X$ and $u_{\overline{X}}$ can be determined as follows.
	
	\begin{itemize}
		\item[] $inputs(1)$ : $u_X$ is either right or North of $u_{\overline{X}}$.
		\item[] $inputs(2)$ : $u_X$ is either left or South of $u_{\overline{X}}$.
		\item[] $inputs(3)$ : $u_X$ is left of $u_{\overline{X}}$.
		\item[] $inputs(4)$ : $u_X$ is right of $u_{\overline{X}}$. $\diamond$
	\end{itemize}
\end{claim}

Suppose that agent $X$ got one of the inputs noted in $inputs(1)$ or $inputs(4)$. 
Agent $X$ chooses action $II$, in view of Table \ref{tab-decision-unknown}.
Since $u_{\overline{X}} \neq pred_{T_{{X}}}(u_{{X}})$, $u_X$ is either right or North of $u_{\overline{X}}$, in view of Claim \ref{claim-relative-positions-of-hitting}.
As a result, $u_{\overline{X}}$ can either be left or South of $u_X$; in other words, $u_{\overline{X}}$ can neither be right of $u_X$ nor North of $u_X$.
Since $u_{{X}}\neq pred_{T_{\overline{X}}}(u_{\overline{X}})$, neither $inputs(1)$ nor $inputs(4)$ apply to agent $\overline{X}$, in view of Claim \ref{claim-relative-positions-of-hitting}.
In view of Claim \ref{claim-num-neighbors}, agent $\overline{X}$ can only get one of the inputs noted in Table \ref{tab-decision-unknown}.
Regardless of whether agent $\overline{X}$ gets $input(2)$ or $input(3)$, it always chooses action $I$.
Therefore, agents $X$ and $\overline{X}$ choose different actions.

Suppose that agent $X$ got one of the inputs noted in $inputs(2)$ or $inputs(3)$. 
Agent $X$ chooses action $I$, in view of Table \ref{tab-decision-unknown}.
Since $u_{\overline{X}} \neq pred_{T_{{X}}}(u_{{X}})$, $u_X$ is either left or South of $u_{\overline{X}}$, in view of Claim \ref{claim-relative-positions-of-hitting}.
As a result, $u_{\overline{X}}$ can either be right or North of $u_X$; in other words, $u_{\overline{X}}$ can neither be left of $u_X$ nor South of $u_X$.
Since $u_{{X}}\neq pred_{T_{\overline{X}}}(u_{\overline{X}})$, neither $inputs(2)$ nor $inputs(3)$ apply to agent $\overline{X}$, in view of Claim \ref{claim-relative-positions-of-hitting}.
In view of Claim \ref{claim-num-neighbors}, agent $\overline{X}$ can only get one of the inputs noted in Table \ref{tab-decision-unknown}.
Regardless of whether agent $\overline{X}$ gets $input(1)$ or $input(4)$, it always chooses action $II$.
Therefore, agents $X$ and $\overline{X}$ choose different actions.

	Therefore, agents $X$ and $\overline{X}$ always choose different actions.
	This proves the lemma.
\end{proof}

As noted in the beginning of this section, if agents choose different actions in the third part of the algorithm then rendezvous occurs.
Hence, 
Lemmas \ref{theorem-ends-in-2nd-part} and  \ref{theorem-no-pred-diff-actions} imply the following theorem
that proves the correctness of Algorithm {\tt Hardest Scenario}.
\begin{theorem}
Algorithm {\tt Hardest Scenario} guarantees rendezvous even if the agents do not have any a priori knowledge and start with arbitrary delay.
\end{theorem}

We conclude this section by analyzing the complexity of the algorithm.

\begin{theorem}\label{theorem-unknow}
	Suppose that the agents are placed at two nodes of the grid at a distance at most $D$.
	Then Algorithm {\tt Hardest Scenario} guarantees rendezvous in time $O(D^2)$.
\end{theorem}
\begin{proof}
	All three parts of Algorithm {\tt Hardest Scenario} are executed by agents sequentially.
	To analyze the worst case running time, we may assume that rendezvous does not occur until both agents start executing the third part of the algorithm.
	We begin with the following claim.
		\begin{claim}
		\label{claim-lenght-of-Trajectory}
		For any $X\in\{B, R\}$, we have
 	 $|T_X| \leq 4D(D+1)$,  $T_X.height\leq 2D$ and $T_X.width\leq 2D$. 
	\end{claim}
	
	To prove the claim, let $R(v, \delta)$, for any even integer $\delta$, denote the set of nodes on the boundary of the square with center at node $v$ and side-length (i.e., the number of edges) $\delta$.
	In particular, the set $R(v, 0)$ is a singleton that contains only node $v$, and the set $R(v, \delta)$ is empty if $\delta<0$.
	The following observations are crucial:
	\begin{itemize}
\item The set $R(v, \delta)$ contains $\max(1, 4\delta)$ nodes.
\item The set $R(v, \delta_1) \cap R(v, \delta_2)$ is empty, if $\delta_1\ne \delta_2$.
\item The four neighbors of any node in $R(v, \delta)$ are in the set $R(v, \delta-2) \cup R(v, \delta)\cup R(v, \delta+2)$.
\item Any node that is at a distance at most $\ell$ from $v$ is in the set $\cup_{i=0}^{\ell} R(v, 2i)$, for any integer $\ell\ge 0$.
	\end{itemize}
	We only prove the last observation, as all the other observations are straightforward.
	The statement is straightforward for $\ell=0$.
	Next, we assume that the statement is true for $\ell=k$ and prove it for $\ell=k+1$.
	Let $s$ denote any node that is at distance $(k+1)$ from $v$, let $\pi$ denote any one of the shortest paths from $v$ to $s$ and let $s'$ denote the neighbor of $s$, that is on $\pi$.
	Since the distance between $v$ and $s'$ is exactly $k$, it follows that $s'\in \cup_{i=0}^{k} R(v, 2i)$, in view of the inductive hypothesis.
	Let $k'$ denote the integer such that $s'\in R(v, 2k')$.
	Obviously, $k'\le k$, and in view of the second observation, $k'$ is unique. 
	Node $s$, as a neighbor of $s'$, is in the set $R(v, k'-2) \cup R(v, k')\cup R(v, k'+2)$, in view of the third observation.
	Hence, $s$ is in the set  $R(v, k'+2) \cup \cup_{i=0}^{k} R(v, 2i)$, which is a subset of $\cup_{i=0}^{k+1} R(v, 2i)$.
	This concludes the proof of the last observation by induction.
	
	Let $b_X$ and $b_{\overline{X}}$ denote the bases of agents $X$ and $\overline{X}$, respectively.
	In view of Procedure {\tt Squarespiral}, agent $X$, after being woken up by the adversary, visits sequentially nodes in $R(b_X, 0)$, $R(b_X, 2)$, $R(b_X, 4)$, $R(b_X, 6)$, and so on, until it makes the first hit.
	Since the distance between $b_X$ and $b_{\overline{X}}$ is at most $D$, node $b_{\overline{X}}$ is in the set $\cup_{i=0}^D R(b_X, 2i)$, in view of the last observation.
	By the round in which agent $X$ completes visiting for the first time all the nodes in $R(b_X, 2D)$, it has already made a hit, either at node $b_{\overline{X}}$ or before visiting $b_{\overline{X}}$.
	Therefore, we have $T_X.height\le 2D$ and $T_X.width\le 2D$.
	In view of the first observation, $|T_X|$, i.e., the number of edges in $T_X$, is at most
	$$|\cup_{i=0}^D R(b_X, 2i)|-1=4\times 2+4\times 4+4\times6+\cdots+4\times 2D=8\times\sum_{i=1}^{D} i=4D(D+1).$$
	This proves the claim. $\diamond$

%
	Next, we analyze the running times of Procedures {\tt Squarespiral} and {\tt PredecessorProbe}, respectively.
	In each move of Procedure {\tt Squarespiral}, the agent visits a new node, so the number of traversed edges, is
	equal to the number of rounds spent on Procedure {\tt Squarespiral}.
	In view of Claim \ref{claim-lenght-of-Trajectory}, Procedure {\tt Squarespiral} uses at most $4D(D+1)$ rounds.
	Procedure {\tt PredecessorProbe} uses no more than $\max(T_X.height, T_X.width)+2$ rounds, which is bounded by $2D+2$, in view of Claim \ref{claim-lenght-of-Trajectory}.
	
	Let $t$ denote the earliest round by which both agents have already started the third part of the algorithm.
	It follows that $t\le 4D(D+1)+2D+2$.
	Without loss of generality, assume that agent $X$ chooses action $I$ and agent $\overline{X}$ chooses action $II$ in the third part of the algorithm.
	By round $t$, agent $\overline{X}$ has visited node $u_{\overline{X}}$ and stays there forever.
Since $u_{\overline{X}}$ is domestic for agent $X$, it is on the trajectory $T_X$.
After agent  $\overline{X}$ stabilizes at $u_{\overline{X}}$, at most one execution of Procedure {\tt Swing} performed by agent $X$ guarantees that rendezvous occurs.
	Hence, agent $X$ must find agent $\overline{X}$ at node $u_{\overline{X}}$ by round $t+2|T_X|\le12D^2+14D+2$.
	Therefore, the time of Algorithm {\tt Hardest Scenario} is at most $12D^2+14D+2 \in O(D^2)$.
\end{proof}

\section{Discussion}

In our model, we made the following two assumptions. We assumed that the grid is oriented, i.e., that the port numbers at each node are  labeled $N, E, S, W$ in a coherent way. This is equivalent to assuming that the grid is embedded in the plane with all edges going North/South or East/West, and that there are no port numbers at nodes but agents have correct compasses indicating the cardinal directions. The second assumption was that the time of rendezvous is counted starting in the wake-up round of the later agent. Below we discuss these two assumptions.

First, consider the assumption about correct compasses. It is easy to see that if this assumption is removed then rendezvous is impossible in many cases. Consider the grid embedded in the plane with all edges of length 1 and going North/South or East/West, without port numbers. Suppose that one agent has a compass showing North in the East direction and the other has a compass showing North in the West direction. Suppose that the two agents are placed by the adversary at two nodes $u$ and $v$ at distance 2 in the grid, not on the same line of the grid. Let $p$ be the center of the segment $[u,v]$. ($p$ is of course not a node of the grid.) The adversary wakes up both agents simultaneously.
It is easy to prove by induction on the round number that, regardless of which deterministic algorithm is applied,  in each round, the agents visit nodes that are central-symmetric with respect to point $p$, and, at the time of the visit 
these nodes are either both marked or both unmarked (the latter condition has to be added to carry out the induction). Hence the agents will never meet.

Next, consider the assumption that the time of rendezvous is counted starting in the wake-up round of the later agent. An alternative would be to count time starting in the wake-up round of the earlier agent. How would this change affect our results? Of course, in the scenario with simultaneous start, there would be no change. Consider the two scenarios with arbitrary delay, and suppose that the initial distance between the agents is at most $D$. Regardless of whether agents know $D$ or not, if the adversary chooses not to wake up the later agent until it is found by the earlier agent, the (worst-case) time of rendezvous must be  $\Omega(D^2)$ because the adversary can choose as the base of the later agent the last node at a distance at most $D$ from the base of the earlier agent, visited by this agent, and there are $\Omega(D^2)$ such nodes. On the other hand, it is easy to see that in our rendezvous algorithm working for the most difficult scenario, the delay between the wake-up of the earlier and later agents is $O(D^2)$ because within this number of rounds the earlier agent visits all nodes at a distance at most $D$ from its base, and hence would catch the later agent if it had not been yet woken up at the time of the visit of its base by the earlier agent. As we proved that our rendezvous algorithm for this most difficult scenario guarantees time $O(D^2)$ since the wakeup of the later agent, the above remark shows that our algorithm
provides a matching upper bound $O(D^2)$ on the complexity of rendezvous if time is counted from the wake-up of the earlier agent.

\section{Conclusion}

We considered three scenarios for deterministic rendezvous of agents leaving traces in the infinite oriented grid.
In the first scenario, agents know an upper bound $D$ on their initial distance but may start with arbitrary delay. In the second scenario, they 
start simultaneously but do not have any {\em a priori} knowledge. In the third scenario, agents start with arbitrary delay and they do not have any 
{\em a priori} knowledge.

For the first two scenarios we provided rendezvous algorithms with complexity linear in $D$ which is obviously optimal. In the third scenario we provided an algorithm with complexity $O(D^2)$ and showed that there is no algorithm with complexity $o(D^{\sqrt{2}})$. While the negative result shows a separation between the optimal complexity in the two easier scenarios from the optimal complexity in the most difficult scenario, finding this optimal complexity is a natural open problem resulting from our research.

Future research on this topic could include studying feasibility and complexity of deterministic rendezvous of anonymous agents leaving traces, in arbitrary connected graphs.
It is easy to see that, in some situations, symmetry cannot be broken and thus rendezvous is impossible. Such is the case, e.g., when agents start simultaneously from antipodal nodes of an oriented ring. Hence the first step of generalizing our considerations to arbitrary graphs should be characterizing the class of graphs for which rendezvous is possible regardless of the initial positions of the agents and regardless of their starting times.

\section{Acknowledgements}

Andrzej Pelc was partially supported by NSERC discovery grant 2018-03899 and by the Research Chair in Distributed Computing at the Universit\'e du Qu\'{e}bec en Outaouais.



\begin{thebibliography}{00}
	
	
	\bibitem{alpern02b}
	S. Alpern and S. Gal,
	The theory of search games and rendezvous.
	Int. Series in Operations research and Management Science,
	Kluwer Academic Publisher, 2002.
	
	\bibitem{BBBG}
	E. Bampas, J. Beauquier, J. Burman, W. Guy--Ob{\'e},
	Treasure hunt with volatile pheromones.
	Proc. 37th International Symposium on Distributed Computing (DISC 2023), 8:1 -8:21.
	
	
	\bibitem{BCGIL}
	E. Bampas, J. Czyzowicz, L. Gasieniec, D. Ilcinkas, A. Labourel, Almost optimal asynchronous rendezvous in infinite multidimensional grids,
	Proc. 24th International Symposium on Distributed Computing (DISC 2010),  297-311.
	
	\bibitem{BP}
	S. Bhagat, A. Pelc, How to meet at a node of any connected graph, 
	Proc. 36th International Symposium on Distributed Computing (DISC 2022), 11:1 - 11:16. 
	
	\bibitem{BP2}
	S. Bhagat, A. Pelc, 
	Deterministic rendezvous in infinite trees, Theoretical Computer Science  984 (2024), 114313.
	
	\bibitem{BBDDP}
	S. Bouchard, M. Bournat, Y. Dieudonn\'{e}, S. Dubois, F. Petit,
	Asynchronous approach in the plane: a deterministic polynomial algorithm, 
	Distributed Computing 32 (2019), 317-337.
	
	\bibitem{BDL}
	S. Bouchard, Y. Dieudonn\'{e}, A. Lamani,
	Byzantine gathering in polynomial time,
	Distributed Computing 35 (2022), 235-263.
	
	\bibitem{CFPS}
	M. Cieliebak, P. Flocchini, G. Prencipe, N. Santoro, 
	Distributed computing by mobile robots: Gathering, SIAM J. Comput. 41 (2012), 829-879.
	
	\bibitem{CCGKM}
	A. Collins, J. Czyzowicz, L. Gasieniec, A. Kosowski, R. A. Martin,
	Synchronous rendezvous for location-aware agents. 
	Proc. 25th International Symposium on Distributed Computing (DISC 2011), 447-459.
	
	\bibitem{CKP}
	J. Czyzowicz, A. Kosowski, A. Pelc, How to meet when you forget: Log-space rendezvous in arbitrary graphs, Distributed Computing 25 (2012), 165-178. 
	
	\bibitem{DP}
	Y.  Dieudonn\'{e}, A. Pelc, Deterministic polynomial approach in the plane, Distributed Computing 28 (2015), 111-129. 
	
	\bibitem{DPV}
	Y. Dieudonn\'{e}, A. Pelc, V. Villain, How to meet asynchronously at polynomial cost, 
	SIAM Journal on Computing 44 (2015), 844-867. 
	
	\bibitem{fpsw}
	P. Flocchini, G. Prencipe, N. Santoro, P. Widmayer,
	Gathering of asynchronous robots with limited visibility, Theoretical Computer Science 337 (2005), 147-168.
	
	\bibitem{KKM}
	E. Kranakis, D. Krizanc, E. Markou,
	Deterministic symmetric rendezvous with tokens in a synchronous torus,
	Discrete Applied Mathematics 159 (2011), 896-923.
	
	\bibitem{KKSS}
	E. Kranakis, D. Krizanc, N. Santoro and C. Sawchuk, 
	Mobile agent rendezvous in a ring, 
	Proc. 23rd Int. Conference on Distributed Computing Systems
	(ICDCS 2003), 592-599.
	
	\bibitem{LR}
	C. Lenzen, T. Radeva,
	The power of pheromones in ant foraging,
	1st Workshop on Biological Distributed Algorithms (BDA), 2013.
	
	
	
	
	\bibitem{MP}
	A. Miller, A. Pelc, Fast deterministic rendezvous in labeled lines, Proc. 37th International Symposium on Distributed Computing (DISC 2023), 29:1 - 29:22. 
	
	\bibitem{Pe2}
	A. Pelc, Deterministic rendezvous algorithms, in: Distributed Computing by Mobile Entities, P. Flocchini, G. Prencipe, N. Santoro, Eds., Springer 2019, LNCS 11340. 
	
	\bibitem{PP}
	D. Pattanayak, A. Pelc,
	Deterministic treasure hunt and rendezvous in arbitrary connected graphs, Information Processing Letters 185 (2024), 106455.
	
	\bibitem{TSZ07}
	A. Ta-Shma and U. Zwick.
	Deterministic rendezvous, treasure hunts and strongly universal exploration sequences,
	ACM Trans. Algorithms 10 (2014), 12:1-12:15.
	
\end{thebibliography}
\end{document}